\documentclass[12pt, onecolumn, draft]{IEEEtran}

\usepackage{cite}
\usepackage[textwidth=6.64in, textheight=24.3cm]{geometry}
\usepackage{amsthm, amsmath,amssymb,amsfonts,bbm}
\usepackage{algorithmic}
\usepackage[dvips]{graphicx}
\usepackage{subfigure}
\graphicspath{{graphics/}, {graphics/Results}, {graphics/Results/0_MatlabResults}}
\usepackage{makecell}
\usepackage{multirow}
\usepackage{filecontents, pgfplots}
\usepackage[algoruled]{algorithm2e}
\usepackage{microtype}
\usepackage{tikz}
\usepackage{textcomp}
\usepackage{xcolor}
\usepackage{verbatim}
\usepackage[binary-units=true]{siunitx}
\DeclareSIUnit{\belmilliwatt}{Bm}
\DeclareSIUnit{\dBm}{\deci\belmilliwatt}
\DeclareSIUnit{\dBi}{dBi}
\def\BibTeX{{\rm B\kern-.05em{\sc i\kern-.025em b}\kern-.08em
		T\kern-.1667em\lower.7ex\hbox{E}\kern-.125emX}}

\allowdisplaybreaks

\makeatletter
\newif\iftag@here

\newcommand*{\taghere}[1][0pt]
{\ifmeasuring@\else
	\global\tag@heretrue
	\tikz[remember picture,overlay]{\coordinate (taghere) at (0pt,#1);}%
	\fi}

\def\place@tag{%
	\iftagsleft@
	\kern-\tagshift@
	\iftag@here
	\global\tag@herefalse
	\tikz[remember picture,overlay]%
	{\path (taghere) -| node[anchor=base]{\rlap{\boxz@}} (0pt,0pt);}%
	\else
	\if1\shift@tag\row@\relax
	\rlap{\vbox{%
			\normalbaselines
			\boxz@
			\vbox to\lineht@{}%
			\raise@tag
	}}%
	\else
	\rlap{\boxz@}%
	\fi
	\kern\displaywidth@
	\fi
	\else
	\kern-\tagshift@
	\iftag@here
	\global\tag@herefalse
	\tikz[remember picture,overlay]%
	{\path  (taghere) -|  node[anchor=base]{\llap{\boxz@}} (0pt,0pt);}%
	\else
	\if1\shift@tag\row@\relax
	\llap{\vtop{%
			\raise@tag
			\normalbaselines
			\setbox\@ne\null
			\dp\@ne\lineht@
			\box\@ne
			\boxz@
	}}%
	\else \llap{\boxz@}%
	\fi
	\fi
	\fi
}
\makeatother

\DeclareMathOperator*{\maximize}{maximize}

\DeclareMathOperator*{\subjectto}{subject\,to}

\usepackage[acronym,nomain]{glossaries}
\makeglossaries
\newacronym{swipt}{SWIPT}{simultaneous wireless information and power transfer}
\newacronym{wpt}{WPT}{wireless power transfer}
\newacronym{wpcn}{WPCN}{wireless powered communication network}
\glsunset{wpcn}
\newacronym{wit}{WIT}{wireless information transfer}
\newacronym{awgn}{AWGN}{additive white Gaussian noise}
\newacronym{tx}{TX}{transmitter}
\newacronym{rx}{RX}{receiver}
\newacronym{bs}{BS}{base station}
\newacronym{ir}{IR}{information receiver}
\newacronym{eh}{EH}{energy harvesting}
\newacronym{id}{ID}{information detection}
\newacronym{irs}{IRS}{intelligent reflected surface}
\newacronym{ap}{AP}{average power}
\newacronym{pp}{PP}{peak power}
\newacronym{aoa}{AOA}{angle-of-arrival}
\newacronym{aod}{AOD}{angle-of-departure}
\newacronym{eec}{EEC}{electical equivalent circuit}

\newacronym{fso}{FSO}{free space optical}
\glsunset{fso}
\newacronym{vlc}{VLC}{visible light communication}
\newacronym{led}{LED}{light emitting diode}
\glsunset{led}
\newacronym{ook}{OOK}{On-Off keying}
\newacronym{ask}{ASK}{amplitude shift keying}
\newacronym{pam}{PAM}{pulse amplitude modulated}
\newacronym{tdma}{TDMA}{time division multiple access}
\newacronym{ml}{ML}{maximum likelihood}

\newacronym{siso}{SISO}{single-input single-output}
\newacronym{mimo}{MIMO}{multiple-input multiple-output}
\newacronym{miso}{MISO}{multiple-input single-output}
\newacronym{simo}{SIMO}{single-input multiple-output}

\newacronym{rf}{RF}{radio frequency}
\newacronym{ghz}{GHz}{gigahertz}
\newacronym{los}{LoS}{line-of-sight}
\newacronym{dc}{DC}{direct current}
\newacronym{ac}{AC}{alternating current}
\newacronym{papr}{PAPR}{peak-to-average power ratio}
\newacronym{lp}{LPF}{low-pass filter}
\newacronym{mc}{MC}{matching circuit}
\newacronym{rtd}{RTD}{resonant-tunelling diode}
\newacronym{mrt}{MRT}{maximum ratio transmission}
\newacronym{ecb}{ECB}{equivalent complex baseband}
\newacronym{tdd}{TDD}{time-division-duplex}
\newacronym{zf}{ZF}{zero forcing}
\newacronym{snr}{SNR}{signal-to-noise ratio}
\newacronym{sinr}{SINR}{signal-to-interference-plus-noise ratio}

\newacronym{rv}{RV}{random variable}
\newacronym{iid}{i.i.d.}{independent and identically distributed}
\newacronym{Pdf}{Pdf}{Probability density function}
\newacronym{pdf}{pdf}{probability density function}
\newacronym{cdf}{cdf}{cumulative density function}

\newacronym{dnn}{DNN}{dense neural networks}
\glsunset{dnn}
\newacronym{mdp}{MDP}{Markov decision process}

\newacronym{sca}{SCA}{successive convex approximation}
\newacronym{sdr}{SDR}{semi-definite relaxation}

\newacronym{spr}{LP}{low power}
\newacronym{mpr}{MP}{medium power}
\newacronym{lpr}{HP}{high power}

\begin{document}

	\newtheorem{proposition}{Proposition}	
	\newtheorem{lemma}{Lemma}	
	\newtheorem{corollary}{Corollary}
	\newtheorem{assumption}{Assumption}	
	\newtheorem{remark}{Remark}	
	
	\title{Information Rate-Harvested Power Tradeoff in THz SWIPT Systems Employing Resonant Tunnelling Diode-based EH Circuits}

	\author{
		Nikita Shanin, 
		Simone Clochiatti, 
		Kenneth M. Mayer, 
		Laura Cottatellucci, 
		Nils Weimann, 
		and Robert Schober
 }

	\maketitle

\begin{abstract}
	\let\thefootnote\relax\footnotetext{This paper was presented in part at the IEEE Globecom Workshops, 2023 \cite{Shanin2023b}.}
\let\thefootnote\relax\footnotetext{This work was supported in part by the German Science Foundation through project SFB 1483 - Project-ID 442419336, EmpkinS.}
\setcounter{footnote}{0}
In this paper, we study terahertz (THz) simultaneous wireless information and power transfer (SWIPT) systems.
{Since coherent information detection is challenging at THz frequencies and Schottky diodes may not be efficient for THz energy harvesting (EH), we propose a novel THz SWIPT system design that employs unipolar amplitude shift keying (ASK) modulation at the transmitter (TX) and a resonant-tunnelling diode (RTD)-based EH circuit at the receiver (RX) to extract both information and power from the received signal. 
Furthermore, we propose a novel model for the dependence of the instantaneous output power of the RTD-based RX on the instantaneous received power, which is based on a non-linear and non-monotonic piecewise function, whose parameters are adjusted to fit circuit simulation results.}
To determine the information rate-harvested power tradeoff of the considered THz SWIPT system, we derive the distribution of the transmit signal that maximizes the \emph{mutual information} between the transmit and received signals subject to constraints on the required average harvested power at the RX and the peak signal amplitude at the TX.
{Since the computational complexity needed for maximization of the {mutual information} may be infeasible for real-time THz SWIPT systems, we derive low-complexity suboptimal input signal distributions that maximize an \emph{achievable information rate} numerically and in closed form for high and low required average harvested powers, respectively.}
Furthermore, based on the obtained results, we propose a suboptimal closed-form distribution of the transmit signal which can also guarantee a desired harvested power at the RX.
{Our simulation results show that while the proposed EH model can capture the non-monotonicity of RTD-based EH circuits in the THz band, baseline linear and non-linear EH models, developed for Schottky-diode-based EH circuits, can not.}
Furthermore, we demonstrate that a lower reverse current flow and a higher breakdown voltage of the employed RTD are preferable when the input signal power at the RX is low and high, respectively.
We also show that all proposed input distributions yield practically identical SWIPT system performance.
Moreover, we reveal that the information rate-harvested power tradeoff of THz SWIPT systems is determined by the peak amplitude of the TX signal and the maximum instantaneous harvested power for low and high received signal powers, respectively.
Finally, we compare the proposed THz SWIPT system with two baseline schemes and confirm that the RX circuit parameters, mathematical EH models, and optimal transmit signal distributions have to be carefully designed to achieve high performance in THz SWIPT systems.
\end{abstract}

\section{Introduction}
\label{Section:Introduction}
{Micro-scale terahertz (THz) Internet-of-Things (IoT) networks are an important use case of future sixth-generation (6G) communication systems \cite{Nguyen2022, Tataria2021}.
In these networks, microscopic low-power IoT devices communicate at extremely high data rates exceeding $\SI{100}{\giga\bit\per\second}$, which are enabled by the huge spectrum available in the THz frequency band \cite{Tataria2021, Kuscu2021, Akyildiz2015}.
Such micro-scale IoT devices, which are also referred to as Internet-of-Nanothings and Internet-of-BioNanothings machines, can be utilized, for example, in healthcare applications for the continuous tracking of patient’s vital signs, in-body cell imaging, or nano-scale surgery \cite{Kuscu2021, Akyildiz2015, Pramanik2020}.
Despite the high expectations for these micro- and nanoscopic IoT devices, the need to regularly replace their batteries is an unsolved problem that presents an obstacle for their practical realization \cite{Sarieddeen2021, Lu2011, Shinohara2021}.
A promising solution to this non-trivial problem is \gls*{swipt}, where both information and power are transmitted in the downlink making battery replacement at the user devices unnecessary \cite{Varshney2008, Zhang2013, Rong2017, Pan2022, Kim2022, Clerckx2019, Shi2017, Pejoski2018, Boshkovska2015, Boshkovska2018, Zhu2021, Zhu2022a, Boshkovska2017a, Xu2022, Zhu2022, Morsi2019, Hanif2023, Shanin2020}.}

The theoretical limits of SWIPT were first analyzed in \cite{Varshney2008}, where the author showed that, for discrete-time SWIPT systems, there exists a fundamental tradeoff between the information rate and the power harvested at a user device.
This tradeoff can be characterized by an information rate-harvested power region, whose boundary is determined by the input signal distributions maximizing the information rate for given required average harvested powers at the user device.
The information rate-harvested power tradeoff of practical SWIPT systems was studied in \cite{Zhang2013}, where the authors considered time-sharing and high-frequency power-splitting between \gls*{eh} and information detection at the \gls*{rx}.
{Furthermore, since the number of scatterers in THz wireless channels is limited \cite{Tataria2021} and a \gls*{los} between \gls*{tx} and \gls*{rx} may not always be available \cite{Sarieddeen2021}, the authors in \cite{Rong2017} and \cite{Pan2022} investigated optimal resource allocation for THz SWIPT systems employing a relay node and intelligent reflecting surfaces (IRSs) to bypass blockages, respectively.
}

For the design of SWIPT systems, the authors  in \cite{Varshney2008, Zhang2013, Rong2017, Pan2022} assumed a linear relationship between the received and the harvested powers.
However, the experimental results in \cite{Kim2022} revealed that practical electrical circuits utilized for \gls*{eh} exhibit a highly non-linear behaviour.
In fact, due to the non-linear current-voltage (I-V) characteristic of EH rectifying diodes, the power harvested by EH circuits can not be accurately characterized by a linear function of the received signal power \cite{Kim2022, Clerckx2019}.
Furthermore, when the input power at the EH device is high, the EH circuit is driven into saturation due to the breakdown of the employed diode \cite{Clerckx2019}.
{To take the non-linear behaviour of EH circuits into account, the authors of \cite{Shi2017, Pejoski2018, Boshkovska2015, Boshkovska2018, Zhu2021, Zhu2022a, Boshkovska2017a, Xu2022, Zhu2022, Morsi2019, Hanif2023, Shanin2020} derived non-linear EH models for the design of EH-based communication networks.
In particular, the authors in \cite{Shi2017} and \cite{Pejoski2018} proposed EH models defined by piecewise linear functions to describe the dependence between the \textit{average} harvested power and the \textit{average} received power of an EH node.
Furthermore, since the input-output characteristic of practical electrical circuits is typically continuous and smooth \cite{Tietze2012, Horowitz1989}, the authors of \cite{Boshkovska2015} proposed an EH model based on a parametric sigmoidal function to characterize the average harvested power at the EH circuit.
The parameters of the EH model in \cite{Boshkovska2015} were adjusted to fit circuit simulation results assuming a Gaussian waveform for the transmit signal.
The sigmoidal EH model in \cite{Boshkovska2015} is widely utilized for the design of EH-based communication networks, comprising also, e.g., eavesdroppers \cite{Boshkovska2018, Zhu2021, Zhu2022a} and relays \cite{Boshkovska2017a}.
In particular, the authors in \cite{Zhu2021} and \cite{Zhu2022a} designed secure IoT SWIPT networks that are robust to imperfect channel knowledge and operate in the mmWave frequency band, respectively.
Moreover, similar to \cite{Pan2022}, the authors of \cite{Xu2022} and \cite{Zhu2022} utilized the EH model in \cite{Boshkovska2015} to investigate IRS-aided gigahertz (GHz)-band and THz-band SWIPT systems, respectively.}

{We note that for the EH models proposed in \cite{Shi2017, Pejoski2018, Boshkovska2015, Boshkovska2018, Zhu2021, Zhu2022a, Boshkovska2017a, Xu2022, Zhu2022}, a \gls*{tx} signal with fixed and known waveform was assumed. 
To optimize the \gls*{tx} waveform for SWIPT, the authors of \cite{Morsi2019} and \cite{Hanif2023} analyzed EH circuits based on Schottky diodes and derived closed-form EH models to characterize the \textit{instantaneous} power harvested at the user device as a function of the \textit{instantaneous} received signal power.
Exploiting the circuit-based EH model, the authors of \cite{Morsi2019} studied SWIPT systems with separate energy and information RXs, and similar to \cite{Varshney2008}, determined the optimal transmit signal distribution that maximizes the mutual information between the TX signal and the signal received at the information RX under a constraint on the average harvested power at the EH RX.
Furthermore, the authors of \cite{Hanif2023} derived an EH model for THz-band imaging systems employing SWIPT assuming On-Off Keying modulation at the TX.}
Finally, since practical IoT devices may require both information and power, the authors of \cite{Shanin2020} determined the optimal input distribution for SWIPT systems, where the EH RX and the information RX were collocated in the same device.
In contrast to \cite{Zhang2013}, the information and energy RXs in \cite{Shanin2020} were connected to different antennas, and thus, could operate simultaneously without splitting the high-frequency received signal.

{The EH models in \cite{Morsi2019, Hanif2023, Shanin2020} characterize the instantaneous harvested power of EH circuits equipped with Schottky diodes.
However, Schottky diodes are not efficient at THz frequencies \cite{Shinohara2021}.
First, the current-voltage (I-V) characteristic of Schottky diodes has relatively low curvature at the zero-bias operating point \cite{Clerckx2019, Tietze2012}.
This property limits the efficiency of Schottky diodes at THz frequencies due to the expected low received signal powers \cite{Sarieddeen2021, Serghiou2022, Priebe2011}.
Furthermore, Schottky diodes typically have large reverse recovery times of around $\SI{100}{\pico\second}$ which limits their ability to efficiently rectify THz signals.
In contrast, resonant tunnelling diodes (RTDs) featuring a potential well and multiple quantum barriers show high efficiency also for low received signal powers and have recovery times of less than $\SI{1}{\pico\second}$, which allows them to efficiently operate at THz frequencies \cite{Clochiatti2022}.}
RTDs typically have a smaller form factor and a larger I-V curvature compared to Schottky diodes, and thus, can lend themselves to integration with THz on-chip antennas and application in microscopic IoT devices \cite{Villani2021}.
The experiments in \cite{Clochiatti2022} revealed that due to an added tunnelling current in a narrow biasing window, the I-V forward-bias characteristic of RTDs is not only highly non-linear, but also exhibits multiple critical points, a region of negative resistance, and a dependence on the signal frequency \cite{Villani2021}.
Moreover, when driven into the region of negative resistance, the RTD may self-oscillate and act as a resonant self-mixing device \cite{Villani2021}.
{Thus, the behaviour of RTDs differs substantially from that of Schottky diodes which exhibit a monotonic I-V characteristic \cite{Tietze2012}.
Hence, the EH models derived for GHz EH circuits equipped with Schottky diodes in \cite{Morsi2019, Hanif2023, Shanin2020} can not accurately characterize RTD-based EH circuits.}
Furthermore, the quality of signal rectification by EH circuits depends not only on the forward-bias I-V characteristic of the employed diode, but also on the diode breakdown voltage and reverse leakage current, which are different for RTDs and Schottky diodes \cite{Horowitz1989, Clochiatti2022}.
Although a complete model capturing the frequency dependence, non-linearity, and non-monotonicity of the current flow through RTDs is not available yet, in \cite{Clochiatti2022}, the authors developed a compact RTD model for Keysight ADS \cite{ADS2017} circuit simulation tool that is based on spline approximation and fits the measurement data presented in \cite{Clochiatti2022} for a wide range of operating frequencies.
{However, to the best of the authors' knowledge, an accurate model for EH circuits and information receivers employing RTDs has not been reported in the literature, yet.}

{Another important aspect of THz SWIPT systems is that, unlike SWIPT systems operating in the GHz frequency band, the design of coherent information RXs, which detect the phase of the received THz signal, is challenging due to the instability and phase noise of THz local oscillators and inefficiency of THz power dividers \cite{Yi2021}.
However, since the spectrum available in the THz band is significantly larger compared to that in the GHz frequency band, high data rates can also be achieved with unipolar \gls*{ask} modulation, which allows for non-coherent detection \cite{Yi2021, Hanif2023}. 
As a result, since EH circuits can also be regarded as envelope detectors, they can be used not only for EH, but also to extract information from the received THz signal and hence, inefficient high-frequency oscillators and THz power dividers at the RX can be avoided.}

In this work, we investigate the information rate-harvested power region and derive TX signal distributions maximizing the mutual information and achievable rate for THz SWIPT systems taking into account all non-linear non-monotonic characteristics of the employed RTD-based RX circuits.
{To the best of the authors' knowledge, this is the first work studying THz SWIPT systems employing efficient, high-speed RTDs for EH.}
{In contrast to other works on THz SWIPT, we introduce a single-user THz SWIPT system employing unipolar ASK modulation at the TX and an RTD-based EH circuit at the RX to extract both information and power from the received THz signal, such that inefficient THz power dividers and oscillators are avoided at the RX.}
The main contributions of this paper can be summarized as follows.
\begin{itemize}
	\item To characterize the instantaneous power of the output signal at an RTD-based RX, we propose a general non-linear non-monotonic EH model.
	The proposed parametric EH model consists of pieces of general 5-parameter logistic functions \cite{Gottschalk2005}, whose parameters are adjusted to fit circuit simulation results and thus, to accurately characterize the instantaneous harvested power at the RXs.
	Thus, it is substantially different from the existing piecewise linear and non-linear EH models derived for the average harvested power \cite{Shi2017, Pejoski2018, Boshkovska2015, Boshkovska2018, Zhu2021, Zhu2022a, Boshkovska2017a, Xu2022, Zhu2022}.
	{Furthermore, we complement the study in \cite{Shanin2023b}, which is the conference version of this paper, by considering not only the RTD developed in \cite{Clochiatti2022} but also two improved RTD designs with a lower reverse leakage current and a higher breakdown voltage, respectively, and by analyzing their impact on the performance of SWIPT systems.}
	\item {To determine the information rate-harvested power region of THz SWIPT systems, we formulate an optimization problem for the maximization of the mutual information between the TX and RX signals subject to constraints on the average harvested power at the RTD-based non-monotonic RX and the peak amplitude of the TX signal.
	We determine a feasibility condition for the problem, and unlike \cite{Shanin2023b}, we derive an algorithm that solves the formulated problem and provides the optimal transmit signal distribution with polynomial time complexity.}
	\item {Since the computational complexity for the maximization of the mutual information may be too high for a practical implementation, we derive the maximum achievable information rate and the corresponding low-complexity input signal distribution numerically and in closed form for high and low required average harvested powers, respectively.
	Additionally, we show that the maximum achievable information rate depends on the ratio between the minimum required average harvested power and the maximum instantaneous harvested power at the RX, and based on this observation, in contrast to \cite{Shanin2023b}, we propose a closed-form input distribution, which yields near-optimal performance.}
	\item {Our simulation results reveal that, in contrast to the proposed EH model, a baseline linear EH model and a baseline non-linear EH model, which was developed for EH circuits equipped with Schottky diodes operating in the GHz band in \cite{Morsi2019}, are not able to accurately capture the non-linear non-monotonic behavior of EH circuits employing RTDs in the THz band.}
	Moreover, we demonstrate that RTD designs with improved reverse leakage current and breakdown voltage are preferable when the received signal power is low and high, respectively.
	Furthermore, we show that the derived optimal and all proposed low-complexity suboptimal input signal distributions yield practically identical achievable information rates and mutual information between TX and RX.
	We illustrate the information rate-harvested power tradeoff in THz SWIPT systems and show that for low and high received signal powers, this tradeoff is determined by the peak amplitude of the transmit signal and the maximum instantaneous harvested power, respectively.
	Finally, we compare the proposed system design with two baseline schemes based on truncated Gaussian signals and optimal distribution for a linear EH model, respectively, and conclude that the transmit signal distribution, RX design, and EH model have to be carefully jointly developed for efficient THz SWIPT.
\end{itemize}

The rest of this paper is organized as follows. 
We present the adopted system model in Section II.
In Section III, we propose a novel parametric non-linear non-monotonic EH model for RTD-based RXs.
In Section IV, we derive novel input signal distributions for THz SWIPT.
In Section V, we present extensive numerical results.
Finally, in Section V, we draw our conclusions.

Throughout this paper, we use the following {\textit{notations}}.
We denote the sets of real, real non-negative, and non-negative integer numbers as $\mathbb{R}$, $\mathbb{R}_{+}$, and $\mathbb{N}$, respectively.
The real-part of a complex variable $x$ is denoted by $\mathcal{R} \{x\}$, whereas $j = \sqrt{-1}$ is the imaginary unit. 
Functions $f_s(s)$ and $F_s(s)$ denote the \gls*{pdf} and \gls*{cdf} of random variable $s$, respectively.
Furthermore, $f(x;y)$ denotes a function of variable $x$ parametrized by $y$, $\delta(\cdot)$ is the Dirac delta function, and $\triangleq$ means ``is defined as".
$\mathbb{E}_s \{\cdot\}$ stands for statistical expectation with respect to random variable $s$.
The domain and first-order derivative of one-dimensional function $g(\cdot)$ are denoted by $\mathcal{D}\{g\}$ and $g'(\cdot)$, respectively.


\section{System Model}
\label{Section:SysModel}
\begin{figure}[!t]
	\centering
	\includegraphics[draft = false, width = 0.5\textwidth]{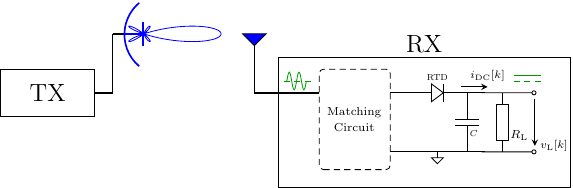}
	\caption{THz SWIPT system model employing a TX with a highly directive antenna and a single-antenna RX equipped with an \gls*{rtd}-based rectifying circuit.}
	\label{Fig:SystemModel}
\end{figure}

We consider the \gls*{swipt} system shown in Fig.~\ref{Fig:SystemModel} where a \gls*{tx} equipped with a highly directive antenna sends a THz signal to a single-antenna \gls*{rx}.
Since coherent demodulation of a THz received signal is challenging, the TX employs unipolar \gls*{ask} modulation for SWIPT \cite{Yi2021, Lemic2021}.
Thus, the THz signal $r(t) \in \mathbb{R}$ received at the RX can be expressed as follows \cite{Horowitz1989}:
\begin{equation}
	r(t) = \sqrt{2} \mathcal{R} \{[h s(t) + n(t)] \exp(j2\pi f_c t) \}, 
\end{equation}
\noindent where $h$ is the channel gain between TX and RX, which is assumed to be perfectly known\footnotemark\hspace*{0pt} at both devices, $f_c$ is the carrier frequency, and $s(t) = \sum_k s[k] \phi(t - kT)$ and $n(t)$ are \gls*{ecb} representations of the transmit signal and the \gls*{awgn} at the RX antenna, respectively.
\footnotetext{In this work, we investigate the maximum achievable performance of THz SWIPT systems, and therefore, we assume perfect channel knowledge at TX and RX \cite{Rong2017, Pan2022, Morsi2019}. 
We note that the transmit signal waveform design in this paper could be extended to account for imperfect channel knowledge, see, e.g., \cite{Zhu2021, Boshkovska2018} for robust designs for GHz SWIPT systems. This is an interesting direction for further research but is beyond the scope of this paper.}
Furthermore, $s[k] \in \mathbb{R}_{+}, k \in \mathbb{N},$ are \gls*{iid} realizations of a non-negative random variable $s$ with \gls*{pdf} $f_s(s)$, $\phi(t)$ is a rectangular pulse that takes value $1$, if $t \in [0, T)$, and $0$, otherwise, and $T$ is the duration of a symbol interval.
{The THz TX is connected to a constant power supply, and to avoid signal distortion at the TX due to the non-linearities of THz power amplifiers\footnote{We note that power losses at the TX depend on the TX signal power, may be caused by the non-ideality of TX circuit components, such as, e.g., THz antennas and power amplifies, and may exceed $10-20\,\SI{}{\decibel}$ \cite{Tataria2021, Shinohara2021, Yi2021, Nguyen2022, Li2022a, Kazan2021}. The modelling and design of THz SWIPT systems taking into account all power losses and imperfections of THz TXs is an interesting direction for future work, which is beyond the scope of this paper.} \cite{Morsi2019, Shanin2020}, the maximum transmit signal power is bounded by $A^2$.}
Thus, the support of $f_s$ is confined to the interval $[0,A]$, i.e., $\mathcal{D}\{f_s\} \subseteq [0, A]$.


\section{EH Modelling for RTD-based RXs}
\label{Section:EHModelling}
In this section, we present the RTD-based electrical EH circuit equipped at the RX and propose a general non-linear piecewise EH model to characterize the instantaneous output RX power.
\subsection{Electrical RX Circuit}
\label{Section:ReceivedSignal}
\begin{figure}[!t]
	\centering
	\includegraphics[draft = false, width=0.4\textwidth]{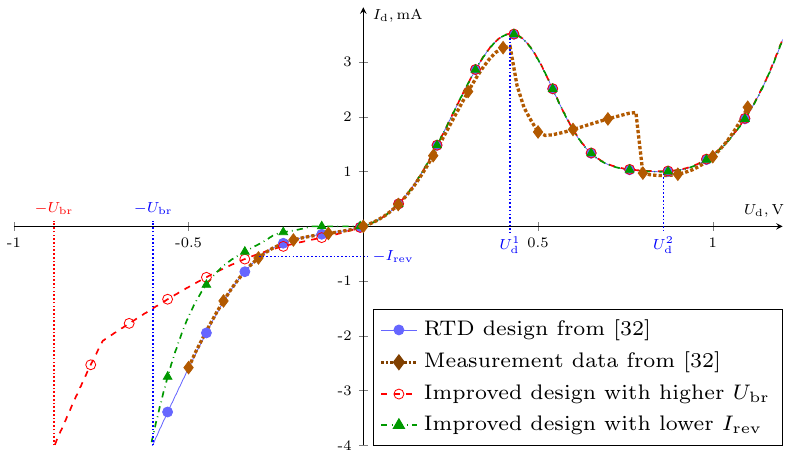}
	\caption{I-V characteristic of the Keysight ADS RTD design in \cite{Clochiatti2022} matched to the measurement data in \cite{Clochiatti2022}, and I-V characteristics of two improved RTD designs with higher breakdown voltage $U_\text{br}$ and lower reverse leakage current $I_\text{rev}$, respectively.}
	\label{Fig:IV_Curve}
\end{figure}

Since Schottky diodes are not efficient for THz signal demodulation and EH \cite{Villani2021, Shinohara2021}, we assume that the RX is equipped with an EH circuit comprising an antenna, a matching circuit, and an RTD-based rectifier with load resistance $R_\text{L}$ \cite{Clerckx2019, Morsi2019, Shanin2020}, as shown in Fig.~\ref{Fig:SystemModel}.
The matching circuit at the RX is a passive device that is designed to match the impedance of the antenna (typically, $\SI{50}{\ohm}$ or $\SI{75}{\ohm}$) to the input impedance of the rectifier \cite{Clerckx2019, Morsi2019, Shanin2020}, which is composed of the RTD, a low-pass filter with capacitance $C$, and a load resistance $R_\text{L}$ \cite{Horowitz1989, Tietze2012}.
{ Practical examples of RX load devices include, e.g., low-power on-body sensors for healthcare \cite{Kuscu2021, Akyildiz2015, Pramanik2020}, sensors for environmental monitoring, or batteries to store the harvested power for future use \cite{Shanin2021a, Cho2020}.}

{In contrast to other works on SWIPT system design \cite{Zhang2013, Rong2017, Pan2022, Kim2022, Clerckx2019, Xu2022, Morsi2019, Zhu2022}, since EH circuits can be regarded as envelope detectors \cite{Tietze2012, Shanin2020}, we utilize the \gls*{dc} voltage $v_\text{L}(t), t \in [(k-1)T, kT),$ generated at the output\footnote{To avoid inefficient high-frequency power dividers, in a practical RX design, one may use a DC current divider \cite{Horowitz1989} to split the output RX current $i_\text{DC}[k]$ in Fig.~\ref{Fig:SystemModel} and utilize portions $\nu i_\text{DC}[k]$ and $(1-\nu) i_\text{DC}[k]$ of the current flow $i_\text{DC}[k], \forall k,$ for EH and information detection, respectively, where $\nu \in [0,1]$ is the current splitting ratio, which depends on the resistivity of the EH load device and information detector. However, since the output signal-to-noise (SNR) ratio and information rate at the information detector do not depend on the ratio $\nu$, we assume that the EH load and information detector are chosen such that $\nu \to 1$.} of the RX circuit in time slot $k, k\in\mathbb{N},$ not only to charge the load, but also to decode information signal $s[k]$, as this avoids the need for inefficient THz oscillators and high-frequency power dividers at the RX.}
Furthermore, as in \cite{Shanin2023b} and \cite{Morsi2019}, we neglect the ripples of the output voltage $v_\text{L}(t)$ across resistance $R_\text{L}$ and the dependence of $v_\text{L}(t)$ in time slot $k, k \in \mathbb{N},$ on information signals $s[p]$ transmitted prior to that time slot, $p < k, p \in \mathbb{N}$, caused by the non-ideality of the employed low-pass filter.
Thus, we express the signal used for information decoding at the RX output, $y[k] \in \mathbb{R}$, in time slot $k, \forall k \in \mathbb{N},$ as follows:
\begin{equation}
	y[k] \triangleq \frac{v_\text{L}[k]}{\sqrt{R_\text{L}}} = \sqrt{\psi( |h s[k]|^2) } + n[k],
	\label{Eqn:OutputRxSignal}
\end{equation}
\noindent where $v_\text{L}[k]$, $n[k]$, and $\psi(\cdot)$ are the output voltage, equivalent output noise sample, and the function that maps the instantaneous received signal power $\rho[k] = | h s[k] |^2$ to the instantaneous power harvested at the load $R_\text{L}$ in time slot $k, \forall k\in\mathbb{N},$ respectively.
Function $\psi(\cdot)$, which models the RTD-based RX circuit in Fig.~\ref{Fig:SystemModel}, will be derived in the next section.
The output noise $n[k]$ in time slot $k, \forall k,$ is composed of the received noise from external sources and the internal thermal noise generated by the components of the RX electrical circuit.
In the following, for the derivation of the average harvested power at the RX, we assume that the impact of noise is negligible \cite{Clerckx2019, Morsi2019, Shanin2020}.
Furthermore, to characterize the performance of information decoding, we assume that the thermal noise originating from the RX components dominates, and thus, we neglect the dependency of $n[k]$ on the received power $\rho[k], \forall k$ \cite{Lapidoth2009}.
Hence, we model the output noise samples $n[k], \forall k,$ as \gls*{iid} realizations of AWGN with zero mean and variance $\sigma^2$.

\subsection{I-V Characteristics of RTDs}
\label{Section:RTDBasedRxCircuit}

In the following, we discuss the RTD employed at the THz RX in Fig.~\ref{Fig:SystemModel}.
We note that RTDs have multiple quantum barriers, and thus, the current flow $I_\text{d}$ through the RTD is not a monotonically increasing function of the applied voltage $U_\text{d}$ \cite{Clochiatti2022}.
Instead, the I-V curve $I_\text{d}(U_\text{d})$ of RTDs may have multiple critical points, as shown in Fig.~\ref{Fig:IV_Curve}.
For simplicity of presentation, we assume a triple-barrier RTD, whose I-V characteristic has a single region where the current flow $I_\text{d}$ decreases when the applied voltage $U_\text{d}$ increases, see Fig.~\ref{Fig:IV_Curve} \cite{Clochiatti2022}.

An ideal diode has zero reverse current flow when the applied voltage is negative, i.e., for $U_\text{d} \in (-\infty, 0]$ \cite{Tietze2012}.
However, due to the presence of minority charge carriers in the semiconductor, the reverse current of a realistic diode is negative, i.e., $I_\text{d} < 0$, even for small negative voltages $U_\text{d} < 0$ \cite{Tietze2012}.
Moreover, if the voltage applied to an RTD is negative and smaller than the breakdown voltage, i.e., $U_\text{d} \leq - U_\text{br}$, the reverse current $I_\text{d}$ grows large, as shown in Fig.~\ref{Fig:IV_Curve}, and may damage the RX \cite{Tietze2012}.

Since the current flow of an RTD is determined by quantum processes and additionally depends on the signal frequency, determining the exact I-V characteristic of RTDs does not seem possible \cite{Clochiatti2022}.
Therefore, for our numerical results in Section~\ref{Section:NumericalResults}, we employ the compact Keysight ADS \cite{ADS2017} numerical model of a triple-barrier RTD developed in \cite{Clochiatti2022}, which was shown to fit the measured\footnotemark\hspace*{0pt} I-V characteristic of the diode designed in \cite{Clochiatti2022}, see Fig.~\ref{Fig:IV_Curve}.
Furthermore, since the rectifier performance depends on the negative-bias characteristics of the diode \cite{Horowitz1989}, to investigate the impact of the reverse leakage current and breakdown voltage of the RTD on SWIPT system performance, we also consider two improved RTD designs. 
In particular, by adjusting the parameters of the ADS diode model developed in \cite{Clochiatti2022}, we design two improved RTDs, whose reverse current $I_\text{rev}$ is lower and breakdown voltage $U_\text{br}$ is higher compared to the RTD in \cite{Clochiatti2022}, respectively.
The I-V characteristics of the RTDs with improved $I_\text{rev}$ and $U_\text{br}$ are depicted by a green dash-dotted line and a red dashed line in Fig.~\ref{Fig:IV_Curve}, respectively.
\footnotetext{We note that the discrepancy between the I-V characteristic of the ADS RTD design in \cite{Clochiatti2022} and the measurement data in Fig.~\ref{Fig:IV_Curve} is caused by non-ideality (parasitic oscillations) of the measurement setup used in \cite{Clochiatti2022}.}
\subsection{Proposed EH Model}
\label{Section:ProposedEHModel}

In contrast to Schottky diodes utilized for RF EH in \cite{Morsi2019, Shanin2020}, the I-V characteristic of RTDs is not monotonic, as shown in Fig.~\ref{Fig:IV_Curve}.
Therefore, the instantaneous harvested power $P_\text{h} = \frac{v_\text{L}^2}{R_\text{L}}$ of an RTD-based EH circuit may not be a monotonic non-decreasing function of the instantaneous input power $\rho = |hs|^2$ \cite{Clochiatti2022, Morsi2019}.
Moreover, since a closed-form expression for the current flow $I_\text{d}$ through an RTD is not available, the derivation of an accurate closed-form expression for the harvested power is not feasible \cite{Clochiatti2022}.
Therefore, in the following, we propose a \textit{general} \textit{parametric} \textit{non-linear} EH model to characterize the instantaneous power harvested at the RX.
{We note that, unlike \cite{Boshkovska2015}, where a parametric model for the average harvested power of EH circuits based on Schottky diodes was proposed, here, we develop a parametric model for the instantaneous harvested power of THz EH circuits employing RTDs.}

\paragraph*{General Piecewise EH Model}
We model the dependence of the instantaneous harvested power on the instantaneous received signal power $\rho$ with the following continuous piecewise function $\psi(\rho)$ shown in Fig.~\ref{Fig:EH_Model}:
\begin{equation}
	\psi (\rho) = \begin{cases}
		\varphi_1(\rho), \; &\text{if} \; \rho \in [\rho_0, \rho_1),\\
		\varphi_2(\rho), \; &\text{if} \; \rho \in [\rho_1, \rho_2),\\
		\cdots \\
		\varphi_N(\rho), \; &\text{if} \; \rho \in [\rho_{N-1}, \rho_\text{max}],
	\end{cases}
	\label{Eqn:EHmodel}
\end{equation} 
\noindent where function $\psi(\cdot)$ is defined in the domain $\mathcal{D}\{\psi\} = [\rho_0, \rho_\text{max}]$ and $\rho_\text{max}$ is the maximum instantaneous received signal power that does not drive the RTD into breakdown.
Here, for modelling $\psi(\cdot)$, we use\footnotemark\hspace*{0pt} $N \in \mathbb{N}$ monotonic functions $\varphi_n(\cdot)$ with domains $\mathcal{D}\{\varphi_n\} = [\rho_{n-1}, \rho_n), n \in \{1,2,\cdots, N-1\}$, and $\mathcal{D}\{\varphi_N\} = [\rho_{N-1}, \rho_\text{max}]$, where $0 \triangleq \rho_0 \leq \rho_1 \leq \rho_2 \leq \cdots \leq \rho_N \triangleq \rho_\text{max}.$
\footnotetext{For the RX circuit shown in Fig.~\ref{Fig:SystemModel} with a single triple-barrier RTD, the number of functions required to model $\psi(\cdot)$ does not exceed $N = 3$ \cite{Villani2021, Clochiatti2022}. However, for more complex RX designs with multiple-barrier RTDs, we may need $N > 3$ functions for accurate EH modelling.}
The required number of functions $N$ depends on the number of critical points in the I-V characteristic and the breakdown voltage $U_\text{br}$ of the RTD.
Furthermore, since the I-V characteristic of an RTD includes both regions where $I_\text{d}$ increases when $U_\text{d}$ grows and where $I_\text{d}$ decreases when $U_\text{d}$ grows, as shown in Fig.~\ref{Fig:IV_Curve}, we adopt parametric monotonically increasing and decreasing functions $\varphi_n(\cdot)$ for odd and even values of $n$, i.e., $n \in \{1,3,\cdots\}$ and $n \in \{2,4,\cdots\}$ with $n \leq N$, respectively, as shown in Fig.~\ref{Fig:EH_Model}.
Finally, we express the average harvested power at the RX as function $\bar{P}_{s}(f_s)$ of the input pdf $f_s(\cdot)$, or equivalently, function $\bar{P}_{x}(f_x)$ of the pdf $f_x(\cdot)$ of signal $x = \sqrt{\psi( |h s[k]|^2) }$ in (\ref{Eqn:OutputRxSignal}) as follows:
\begin{equation}
	\bar{P}_{s}(f_s) \triangleq \mathbb{E}_s \{ \psi(|hs|^2) \} = \mathbb{E}_x \{ x^2 \} \triangleq \bar{P}_{x}(f_x). 
	\label{Eqn:AverageHarvestedPower}
\end{equation}
\noindent Here, we neglect the impact of noise since its contribution to the average harvested power is negligible \cite{Boshkovska2015, Morsi2019, Shanin2020}.
\begin{figure*}[!t]
	\centering
	\begin{minipage}{.5\textwidth}
		\flushleft
		\includegraphics[draft=false, width=0.8\textwidth]{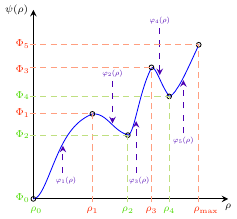}
		\caption{Proposed EH model in (\ref{Eqn:EHmodel}) with $N = 5$ monotonic functions $\varphi_n(\cdot)$.}
		\label{Fig:EH_Model}
	\end{minipage}\hspace*{25pt}
	\begin{minipage}{0.5\textwidth}
		\flushleft
		\includegraphics[draft=false, width=0.8\textwidth]{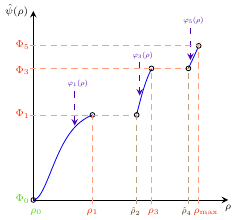}
		\caption{Equivalent EH model for the solution of (\ref{Eqn:GeneralOptimizationProblem}).}
		\label{Fig:EH_ModelEquiv}
	\end{minipage}
\end{figure*}

\paragraph*{\color{black}Parametric model of $\varphi_n(\cdot)$}
{In contrast to \cite{Shi2017, Pejoski2018, Boshkovska2015, Boshkovska2018, Zhu2021, Boshkovska2017a, Xu2022, Zhu2022, Zhu2022a}, an accurate analytical derivation of the non-linear functions $\varphi_n(\cdot), n\in\{1,2,\cdots, N\},$ does not seem feasible for RTD-based EH circuits.
Therefore, we propose to model $\varphi_n(\cdot), \forall n,$ by the following $5$-parameter logistic function \cite{Gottschalk2005}:
\begin{equation}
	\varphi_n(\rho) = B_n + (\Phi_{n-1} - B_n) \Big[ 1 + \big(\theta_n (\rho - \rho_{n-1})\big)^{\alpha_n} \Big]^{-\beta_n}.
	\label{Eqn:SigmoidalFunctionModel}
\end{equation}
\noindent Here, $B_n = \lim_{\rho \to \infty} \varphi_n(\rho)$ and $\Phi_n = \varphi_{n}(\rho_{n}), \forall n,$ with $\Phi_0 = 0$ and $\Phi_n \leq \Phi_{n+2}, \forall n \in \{0,1,\cdots, N-2\}$, as shown in Fig.~\ref{Fig:EH_Model}.
Parameters $\alpha_n, \beta_n, \theta_n, B_n, \rho_{n}  \in \mathbb{R}_+$ in (\ref{Eqn:SigmoidalFunctionModel}) characterize the non-linearity of $\varphi_n(\cdot), \forall n,$ and can be obtained through a curve-fitting approach to maximize the agreement between $\psi(\cdot)$ and measurement or simulation data.
We note that in contrast to \cite{Shi2017, Pejoski2018, Boshkovska2015, Boshkovska2018, Zhu2021, Zhu2022a, Boshkovska2017a, Xu2022, Zhu2022}, where the average harvested power of Schottky diode-based EH circuits was considered, the proposed piecewice non-monotonic EH model is designed to accurately characterize the non-linear non-monotonic instantaneous harvested power at the RTD-based EH circuit.}
Furthermore, as in \cite{Zhang2013, Rong2017, Pan2022, Kim2022, Clerckx2019, Boshkovska2015, Morsi2019, Shanin2020, Xu2022, Boshkovska2018, Boshkovska2017a}, we assume that all parameters of the RX circuit are perfectly known at the TX.


\section{Information Rate-Harvested Power Tradeoff}
\label{Section:RatePowerTradeoffs}
In this section, we determine the optimal input distribution and characterize the information rate-harvested power tradeoff of the considered THz SWIPT system.
To this end, we first formulate and solve an optimization problem for the maximization of the mutual information between TX and RX subject to constraints on the average harvested power at the RX and the peak TX signal power.
Next, since solving the optimization problem is computationally expensive for practical SWIPT systems, we obtain the input signal distribution maximizing the achievable information rate for a given required average harvested power at the RX. 
Finally, based on this result, we propose a closed-form TX signal distribution that also provides a suboptimal solution of the formulated optimization problem.
\subsection{Problem Formulation and Optimal Solution}
\label{Section:ProblemFormulationOptimalSolution}
We characterize the performance of the THz SWIPT system by the tradeoff between the average harvested power at the RX $\bar{P}_{s}(f_s)$ in (\ref{Eqn:AverageHarvestedPower}) and the mutual information $I_s(\cdot)$ between TX and RX that can be expressed in nats per channel use as follows \cite{Morsi2019, Shanin2020, Grover2010}:
\begin{equation}
	I_s(f_s) = -\int_{-\infty}^{\infty} f_y(y; f_x) \ln f_y(y; f_x) \text{d}y - \frac{1}{2} \ln(2\pi e \sigma^2).
	\label{Eqn:MutualInformation}
\end{equation}
Here, $f_y(y; f_x)$ is the pdf of the RX output signal $y = x + n$ in (\ref{Eqn:OutputRxSignal}) as a function of the pdf $f_x$ of output information signal $x = \sqrt{\psi( |h s[k]|^2) }$ and is given by \cite{Shanin2020, Grover2010}:
\begin{equation}
	f_y(y; f_x) = \int_{-\infty}^{\infty} f_x(x; f_s) f_n(y-x)\text{d}x,
	\label{Eqn:OutputPdf}
\end{equation}
\noindent where $f_x(x; f_s)$ is the pdf of the output information signal $x$ for a given input pdf $f_s$ and $f_n(n) = \frac{1}{\sqrt{2\pi \sigma^2}} \exp(-\frac{n^2}{2 \sigma^2})$ is the pdf of the AWGN $n$.

To characterize the information rate-harvested power tradeoff, we determine the pdf $f_s$ of transmit signal $s$ that maximizes the mutual information $I_s(\cdot)$ for a given required average harvested power $\bar{P}_\text{req}$ at the RX \cite{Morsi2019, Varshney2008, Shanin2020}.
To this end, we formulate the following optimization problem:
\begin{subequations}
	\begin{align}
		\maximize_{f_s \in \mathcal{F}_{s}  }\quad \; &I_s(f_s) \label{Eqn:GeneralObj}\\
		\subjectto \quad\; & \bar{P}_{s}(f_s) \geq \bar{P}_\text{req}, \label{Eqn:GeneralOptConstr1}
	\end{align}
	\label{Eqn:GeneralOptimizationProblem}
\end{subequations}
\noindent\hspace*{-4pt}where $\mathcal{F}_{ s } = \{f_s \; \vert \; \mathcal{D}\{f_s\} \subseteq [0, \bar{A}], \int_{s}f_s(s)\, \text{d}s = 1 \}$ denotes the set of feasible input pdfs whose support does not exceed the maximum value $\bar{A} = \min \{A, \frac{ \sqrt{\rho_\text{max}} }{|h|}\}$, such that the transmit signal amplitude is upper-bounded by $A$ and the RTD is not driven into breakdown.

To solve (\ref{Eqn:GeneralOptimizationProblem}), in the following proposition, we first determine under which condition optimization problem (\ref{Eqn:GeneralOptimizationProblem}) is feasible.
\begin{proposition}
	For a given $\bar{A} = \min \{A, \frac{ \sqrt{\rho_\text{\upshape max}} }{|h|}\}$, optimization problem (\ref{Eqn:GeneralOptimizationProblem}) is feasible if and only if $\bar{P}_\text{\upshape req} \in \Big[0, \bar{P}_\text{\upshape max} \Big]$ with $\bar{P}_\text{\upshape max} = \max_{\rho \in [0, |h\bar{A}|^2 ] } \psi(\rho)$.
	\label{Prop:Feasibility}
\end{proposition}
\begin{proof}
	We note that for a given maximum transmit signal amplitude $\bar{A}$, the instantaneous harvested power can not exceed $\bar{P}_\text{\upshape max}$.
	Then, the average power harvested at the RTD-based RX is upper-bounded by $\max_{f_s \in {\mathcal{F}}_{s}} \bar{P}_\text{s}(f_s) = \bar{P}_\text{\upshape max}$.
	Thus, for any $\bar{P}_\text{req} > \bar{P}_\text{\upshape max}$, a solution of (\ref{Eqn:GeneralOptimizationProblem}) does not exist.
	On the other hand, for any $\bar{P}_\text{req} \in[0,  \bar{P}_\text{\upshape max}]$, there exists at least one pdf $f^1_s = \delta(s-s_0) \in \mathcal{F}_{ s }$ that satisfies constraint (\ref{Eqn:GeneralOptConstr1}) with equality, where $s_0 \leq {\bar{A}}$ is chosen such that $\psi(|hs_0|^2) = \bar{P}_\text{req}$.
	This concludes the proof.
\end{proof}

Proposition~\ref{Prop:Feasibility} highlights that for a given peak amplitude ${ \bar{A} }$, the average power harvested at the RX is bounded by $\bar{P}_\text{\upshape max}$, and hence, a solution of (\ref{Eqn:GeneralOptimizationProblem}) does not exist for $\bar{P}_\text{\upshape req} > \bar{P}_\text{\upshape max}$.
Furthermore, we note that if $\bar{P}_\text{\upshape req} = \bar{P}_\text{\upshape max}$, the optimal pdf solving (\ref{Eqn:GeneralOptimizationProblem}) is trivial and given by $f_s^\text{opt}(s) = \delta(s - \bar{A})$.
In the following, we consider the case, where optimization problem (\ref{Eqn:GeneralOptimizationProblem}) is feasible and the solution is not trivial, i.e., $\bar{P}_\text{\upshape req} \in [0, \bar{P}_\text{\upshape max})$.

We note that since function $\psi(\cdot)$ is, in general, not invertible, the mapping between the input pdf $f_s$ and the pdf $f_x$, and hence, mutual information $I_s(\cdot)$ in (\ref{Eqn:MutualInformation}) may not be unique.
Therefore, in the following proposition, we show that to determine an optimal input pdf $f_s^{\text{opt}}$ as solution of (\ref{Eqn:GeneralOptimizationProblem}), the proposed general piecewise EH model in (\ref{Eqn:EHmodel}) can be replaced by an equivalent EH model that is increasing in its domain.
\begin{proposition}
	There exists an optimal input pdf solving (\ref{Eqn:GeneralOptimizationProblem}) that satisfies $f_s^{\text{\upshape opt}}(s) = 0$, $\forall s \in \mathcal{S}$, where $\mathcal{S} = (s_1, \hat{s}_2] \cup (s_3, \hat{s}_4] \cup \cdots \cup (s_{\hat{N}-1}, \hat{s}_{\hat{N}}]$.
	Here, $s_n = \sqrt{ \frac{\rho_n}{|h|^2} }$ and $\hat{s}_n = \sqrt{ \frac{\hat{\rho}_n}{|h|^2} }, n \in \{1, 3, 5, \cdots, \hat{N}\}$.
	Furthermore, $\hat{\rho}_n = \varphi_{n+1}^{-1} \big( \varphi_{n-1}(\rho_{n-1}) \big), n \in \{2, 4, 6, \cdots, \hat{N}-1\}$, and $\hat{N} \in \mathbb{N}$ is the maximum $n\leq N$ such that $\hat{\rho}_{n-1}$ exists.
	\label{Prop:EquivalentModel}
\end{proposition}
\begin{proof}
	Let us consider a distribution with pdf $f_s \in \mathcal{F}_{s}$, which has a mass point at $\tilde{s} \in \mathcal{S}$.
	We note that there exists a point $\bar{s} \notin \mathcal{S}$ and $\bar{s} \leq \tilde{s}$ that yields the same signal $x=\sqrt{\psi( |hs|^2) }$ as $\tilde{s}$, i.e., $\sqrt{\psi( |h\bar{s}|^2) } = \sqrt{\psi( |h\tilde{s}|^2) }$, as shown in Fig.~\ref{Fig:EH_ModelEquiv}.
	Thus, for any such distribution $f_s$, an equal or larger value of $\bar{P}_\text{s}(\cdot)$ and $I_s(\cdot)$ can be attained by removing the mass point $\tilde{s}$ and increasing the probability of $\bar{s}$ by $f_s(\tilde{s})$.
	Since the new probability is in $\mathcal{F}_{s}$, there exists an optimal input pdf $f_s^\text{opt}$ as solution of (\ref{Eqn:GeneralOptimizationProblem}) that does not have mass points in $\mathcal{S}$. This concludes the proof.
\end{proof}

Proposition~\ref{Prop:EquivalentModel} reveals that the solution of problem (\ref{Eqn:GeneralOptimizationProblem}) may not be unique and there exists an optimal input pdf with $f_s^\text{opt} = 0, \forall s \in \mathcal{S}$.
Thus, for any given pdf $f_x$, an input pdf $f_s$ that yields $f_x$ can be obtained as $f_s(s) = \frac{\partial}{\partial s} F_x(\hat{\psi} (|hs|^2))$, where $F_x(x) = \int_{0}^x f_x(\tilde{x}) \text{d}\tilde{x}$ is the corresponding cdf of $x$.
Here, $\hat{\psi}(\cdot)$ is the equivalent EH model with $\hat{\psi}(\rho) = \psi(\rho), \forall \rho \in \mathcal{D}\{ \hat{\psi} \},$ that is shown in Fig.~\ref{Fig:EH_ModelEquiv}.
The domain of function $\hat{\psi}(\cdot)$ is $\mathcal{D}\{ \hat{\psi} \} = [\rho_0, \rho_1) \cup [\hat{\rho}_2, \rho_3) \cup [\hat{\rho}_4, \rho_5) \cup \cdots \cup [\hat{\rho}_{\hat{N}-1}, \rho_{\text{max}}]$.
We note that in contrast to EH model $\psi(\cdot)$ in (\ref{Eqn:EHmodel}), function $\hat{\psi}(\cdot)$ is monotonically increasing in its domain $\mathcal{D}\{\hat{\psi}\}$, and thus, is invertible.
Therefore, for a given $f_x$, the input pdf $f_s$ exists and is unique.

To obtain an optimal pdf $f_s^\text{opt}$ as solution of (\ref{Eqn:GeneralOptimizationProblem}), we can first determine the optimal pdf $f_x^\text{opt}$ of the signal $x$ that solves the following equivalent optimization problem:
\begin{subequations}
	\begin{align}
		\maximize_{f_x \in \mathcal{F}_{x}  }\quad \; &{I}_{x}(f_x) \label{Eqn:RefOutputObj}\\
		\subjectto \quad\; & \bar{P}_{x}(f_x) \geq \bar{P}_\text{req}, \label{Eqn:RefOutputOptConstr1}
	\end{align}
	\label{Eqn:RefOutputOptimizationProblem}
\end{subequations}
\noindent\hspace*{-4pt}where ${I}_{x}(f_x) = -\int_{-\infty}^{\infty} f_y(y; f_x) \ln f_y(y; f_x) \text{d}y - \frac{1}{2} \ln(2\pi e \sigma^2)$ is the mutual information between signals $x$ and $y$ expressed as a function of $f_x$, i.e., we have $I_x(f_x) = I_s(f_s)$ in (\ref{Eqn:MutualInformation}) if the input pdf $f_s$ yields $f_x$.
Here, $\mathcal{F}_{x} = \{f_x \; \vert \; \mathcal{D}\{f_x\} \subseteq [0, \sqrt{\bar{P}_\text{max}}], \int_{x}f_x(x)\, \text{d}x = 1 \}$ denotes the feasible set of pdfs $f_x$ that correspond to the input pdfs $f_s \in \mathcal{F}_{s}$ in (\ref{Eqn:GeneralOptimizationProblem}).

We note that determining the optimal pdf $f_x^\text{opt}$ as solution of (\ref{Eqn:RefOutputOptimizationProblem}) may not be possible in closed form.
However, optimization problem (\ref{Eqn:RefOutputOptimizationProblem}) is convex, and hence, we can solve (\ref{Eqn:RefOutputOptimizationProblem}) numerically by discretizing the feasible set $\mathcal{F}_x$ \cite{Boyd2004}.
To this end, we define a uniform grid $\mathcal{X} = \{x_0, x_1, \cdots, x_{K-1}\}$ of size $K \in \mathbb{N}$, where $x_k = \frac{k}{K-1} \sqrt{\bar{P}_\text{max}}$.
Next, for $\mathcal{F}_x = \mathcal{X}$, we determine the discrete pdf $f_x^\text{opt}$ as solution of (\ref{Eqn:RefOutputOptimizationProblem}) utilizing CVX \cite{Grant2015, Morsi2019, Shanin2020}.
For the obtained pdf $f_x^\text{opt}$, we find the corresponding input pdf $f_s^\text{opt}(s) = \frac{\partial}{\partial s} F_x^\text{opt}(\hat{\psi} (|hs|^2))$, where $F^\text{opt}_x(x) = \int_{0}^x f^\text{opt}_x(\tilde{x}) \text{d}\tilde{x}$ is the corresponding cdf of $x$.

The procedure for determining the input pdf $f_s^\text{opt}$ solving (\ref{Eqn:GeneralOptimizationProblem}) is summarized in \textbf{Algorithm~\ref{Alg:OptimalSolution}}.
{We note that as the sampling grid size $K \to \infty$, the obtained discrete pdfs $f_x^\text{opt}$ and $f_s^\text{opt}$ converge to the optimal solutions of (\ref{Eqn:RefOutputOptimizationProblem}) and (\ref{Eqn:GeneralOptimizationProblem}), respectively \cite{Morsi2019}.
Furthermore, the computational complexity of \textbf{Algorithm~\ref{Alg:OptimalSolution}} is given by $\mathcal{O}(K^{3.5})$ \cite{Nesterov1994, Nocedal2006, Akle2015}, where $\mathcal{O}(\cdot)$ is the big-O notation.
Thus, if $K$ is large, the optimal solution of (\ref{Eqn:GeneralOptimizationProblem}) may entail a high computational complexity, which is not desirable in practical THz SWIPT systems \cite{Tataria2021, Sarieddeen2021}.
Therefore, in the following, as a suboptimal solution of (\ref{Eqn:GeneralOptimizationProblem}), we derive a low-complexity input signal distribution that maximizes the achievable information rate between the TX and RX.}

\begin{algorithm}[!t]	
	\small				
	\SetAlgoNoLine%
	Initialize: Peak signal amplitude $\bar{A}$, required average harvested power $\bar{P}_\text{\upshape req}$, EH model $\psi(\cdot)$, channel coefficient $h$, grid size $K$. \\	
	1. Find $\bar{P}_\text{\upshape max} = \max_{\rho \in [0, |h\bar{A}|^2 ] } \hat{\psi}(\rho)$.\\
	2. Discretize the feasible set $\mathcal{F}_x = \{x_0, x_1, \cdots, x_{K-1}\}$, where $x_k = \frac{k}{K-1} \sqrt{\bar{P}_\text{max}}$, $\forall k$ \\
	3. Determine the optimal pdf $f_x^\text{opt} \in \mathcal{F}_x$ as solution of (\ref{Eqn:RefOutputOptimizationProblem}) and find the corresponding cdf $F_x^\text{opt}$ and pdf $f_s^\text{opt}$ \\
	\textbf{Output:} Optimal input pdf $f_s^\text{opt}$
	\caption{\strut Algorithm for determining the input pdf $f_s^\text{opt}$. }
	\label{Alg:OptimalSolution}
\end{algorithm}

\subsection{Achievable Rate-Power Tradeoff}
\label{Section:AchievableTradeoff}
Since accurately determining an optimal input pdf $f_s^\text{opt}$ that solves (\ref{Eqn:GeneralOptimizationProblem}) is challenging in practice, in the following, for a given $\bar{A}$ and $\bar{P}_\text{\upshape req} \in [0, \bar{P}_\text{\upshape max})$, we obtain a low-complexity suboptimal solution of (\ref{Eqn:GeneralOptimizationProblem}).
To this end, in the following lemma, we first derive a lower bound on the maximum mutual information in (\ref{Eqn:RefOutputOptimizationProblem}).

\begin{lemma}
	For any values of ${ \bar{A} }$ and $\bar{P}_\text{\upshape req} \in \Big[0, \bar{P}_\text{\upshape max} \Big]$, the maximum mutual information as solution of (\ref{Eqn:RefOutputOptimizationProblem}) is lower-bounded by 
	\begin{equation}
		I_x(f_x^\text{\upshape opt}) = \max_{f_x \in \bar{\mathcal{F}}_{ x }} \, I_x(f_x) \geq \max_{f_x \in \bar{\mathcal{F}}_{ x }} \; J_x(f_x) \triangleq J^*_{\bar{\mathcal{F}}_{x} },
		\label{Eqn:MI_LowerBound}
	\end{equation}
	\noindent where $\bar{\mathcal{F}}_{ x } = \{f_x \, \vert \, f_x \in \mathcal{F}_{ x }, \bar{P}_{x}(f_x) \geq \bar{P}_\text{\upshape req} \}$ is the feasible set of problem (\ref{Eqn:RefOutputOptimizationProblem}).
	Furthermore, $J_x(f_x) = \frac{1}{2} \ln\big(1 + \frac{e^{2 h_x(f_x)}}{2\pi e \sigma^2 }\big)$ is an achievable information rate as function of $f_x(\cdot)$ and $h_x(f_x)~=~-\int f_x(x) \ln f(x) \text{\upshape d}x$ is the entropy of random variable $x$ for a given pdf $f_x(\cdot)$.
	\label{Lemma:EPI}
\end{lemma}
\begin{proof}
	The proof follows \cite[Appendix A]{Lapidoth2009}.
	Specifically, we express the maximum mutual information in  (\ref{Eqn:RefOutputOptimizationProblem}) as follows:
	\begin{align}
		I^*_{\bar{\mathcal{F}}_{x } } \triangleq \max_{f_x \in \bar{\mathcal{F}}_{ x }} \; I_x(f_x) = \max_{f_x \in \bar{\mathcal{F}}_{ x }} \; h_y(f_x) - h_{n} \geq \max_{f_x \in \bar{\mathcal{F}}_{ x }} \; \frac{1}{2} \ln \big( e^{2 h_x(f_x)} + e^{2 h_n} \big) - h_{n} =  J^*_{\bar{\mathcal{F}}_{x} },
	\end{align}
	\noindent where $h_{y}(f_x)$ and $h_{n} = \frac{1}{2} \ln (2 \pi e \sigma^2)$ are the differential entropies of $y$ for a given pdf $f_x \in \bar{\mathcal{F}}_{ x }$ and AWGN $n$ in (\ref{Eqn:MutualInformation}), respectively.
	This concludes the proof.
\end{proof}

Lemma~\ref{Lemma:EPI} shows that the maximum mutual information $I_x(f_x^\text{opt})$ as solution of (\ref{Eqn:RefOutputOptimizationProblem}) can be lower-bounded by the maximum achievable information rate $J^*_{\bar{\mathcal{F}}_{x} }$, which, in turn, is obtained as a solution of the optimization problem in (\ref{Eqn:MI_LowerBound}).
In the following, as a suboptimal solution of (\ref{Eqn:RefOutputOptimizationProblem}), for given ${ \bar{A} }$ and $\bar{P}_\text{\upshape req} \in [0, \bar{P}_\text{\upshape max} )$, we determine the pdfs $f^\text{ach}_x$ and $f^\text{ach}_s$ of random variables $x$ and $s$, respectively, that yield the maximum achievable rate  $J^*_{\bar{\mathcal{F}}_{x} }$.
First, in the following proposition, we show that for small required average harvested powers $\bar{P}_\text{req}$, constraint (\ref{Eqn:RefOutputOptConstr1}) in the definition of $\bar{\mathcal{F}}_{ x }$ can be relaxed and the pdf $f_x^\text{ach}$ can be obtained in closed form.
\begin{proposition}
	\label{Prop:UniformDistribution}
	For a given $\bar{A} = \min \{A, \frac{ \sqrt{\rho_\text{\upshape max}} }{|h|}\}$ and required average harvested power satisfying $\bar{P}_\text{\upshape req} \in [0, \frac{1}{3} \bar{P}_\text{\upshape max}],$ the maximum achievable information rate is given by $J^*_{\bar{\mathcal{F}}_{x}} =  \frac{1}{2} \ln\big(1 + \frac{ \bar{P}_\text{\upshape max}}{2\pi e \sigma^2 }\big)$ and the corresponding pdf of $x$ is 
	\begin{equation}
		f^\text{\upshape ach}_x(x) = \frac{1}{\sqrt{\bar{P}_\text{\upshape max}}}, \quad x \in [0, \sqrt{\bar{P}_\text{\upshape max}}].
		\label{Eqn:UniformDistributionX}
	\end{equation}
\end{proposition}
\begin{proof}
	First, we note that if the average power constraint (\ref{Eqn:RefOutputOptConstr1}) in the definition of $\bar{\mathcal{F}}_{x}$ is not present, the differential entropy $h_x(f_x)$, and hence, function $J_x(f_x)$ with $f_x \in \bar{\mathcal{F}}_{x}$ are maximized if the pdf of $x$ is uniform and given by $f_x^\text{ach}$ \cite{Lapidoth2009}.
	Furthermore, in this case, the maximum achievable information rate and the average harvested power can be expressed as $J_x(f_x) = J^*_{\bar{\mathcal{F}}_{x}}$ and $\bar{P}_\text{\upshape req} = \mathbb{E}_x\{x^2\} = \frac{1}{3} \bar{P}_\text{\upshape max}$, respectively.
	Thus, for $\bar{P}_\text{\upshape req} \in [0, \frac{1}{3} \bar{P}_\text{\upshape max}]$, for the maximization of $J_x(\cdot)$, constraint (\ref{Eqn:RefOutputOptConstr1}) can be relaxed and the optimal distribution of $x$ that yields $J^*_{\bar{\mathcal{F}}_{ x }}$ is given by (\ref{Eqn:UniformDistributionX}).
	This concludes the proof.
\end{proof}

Proposition~\ref{Prop:UniformDistribution} reveals that if the required average harvested power $\bar{P}_\text{req}$ is low, there is no tradeoff between the achievable information rate and the average harvested power and the corresponding $J^*_{\bar{\mathcal{F}}_{ x }}$, $f_x^\text{ach}$, and thus, $f_s^\text{ach}$ can be computed in closed form.
In the following, we derive the output pdf $f_y^\text{ach}$ in (\ref{Eqn:OutputPdf}) when the pdf of $x$ is given by (\ref{Eqn:UniformDistributionX}). 
\begin{corollary}
	\label{Corollary:Uniform}
	If random variable $x$ has pdf $f^\text{\upshape ach}_x$ in (\ref{Eqn:UniformDistributionX}), the output pdf in (\ref{Eqn:OutputPdf}) is given by
	\begin{align}
		f^\text{\upshape ach}_y(y; f_x^\text{\upshape ach}) &= \frac{1}{\sqrt{\bar{P}_\text{\upshape max}}} \Big[ Q\Big(\frac{y-\sqrt{\bar{P}_\text{\upshape max}}}{\sigma}\Big) - Q\Big(\frac{y}{\sigma}\Big) \Big].
		\label{Eqn:AchOutputPdf_Uniform}
	\end{align}
\end{corollary}
\begin{proof}
	First, we express the pdf $f_y^\text{ach}$ of the output signal $y = x+n$ in (\ref{Eqn:OutputPdf}) as follows:
	\begin{align}
		f_y^\text{ach}(y; f_x^\text{ach}) =\int_{-\infty}^{\infty} f^\text{ach}_x(x) f_n(y-x) \text{d} x = \int_{0}^{ \sqrt{\bar{P}_\text{\upshape max}} } \frac{1}{\sqrt{\bar{P}_\text{max}}} f_n(y-x) \text{d} x.
		\label{Eqn:Prop5Eqn1}
	\end{align}
	Since $n$ is a zero-mean Gaussian random variable with variance $\sigma^2$, we have $\int_{0}^{\sqrt{\bar{P}_\text{\upshape max}}} f_n(y-x) \text{d} x = Q\big(\frac{y-\sqrt{\bar{P}_\text{\upshape max}}}{\sigma}\big) - Q\big(\frac{y}{\sigma}\big) $.
	This concludes the proof.
\end{proof}

In the next proposition, we consider the case where $\bar{P}_\text{req} \in [\frac{1}{3} \bar{P}_\text{\upshape max}, \bar{P}_\text{\upshape max})$ and characterize the corresponding maximum achievable information rate $J^*_{\bar{\mathcal{F}}_{ x }}$ and $f_x^\text{ach}$.
\begin{proposition}
	For a given $\bar{A} = \min \{A, \frac{ \sqrt{\rho_\text{\upshape max}} }{|h|}\}$ and a required average harvested power $\bar{P}_\text{\upshape req} \in [\frac{1}{3} \bar{P}_\text{\upshape max}, \bar{P}_\text{\upshape max})$, the maximum achievable information rate $J^*_{\bar{\mathcal{F}}_{ x }}$ in (\ref{Eqn:MI_LowerBound}) is given by 
	\begin{equation}
		J^*_{\bar{\mathcal{F}}_{ x }} = \frac{1}{2} \ln\big(1 + \frac{e^{2 \mu_0 - 2\mu_1^2 \frac{\bar{P}_\text{\upshape req}}{\bar{P}_\text{\upshape max}} }}{2\pi e \sigma^2 }\big).
		\label{Eqn:OptMI}
	\end{equation}
	\noindent Furthermore, the corresponding pdf of $x$ that solves (\ref{Eqn:RefOutputOptimizationProblem}) can be expressed as
	\begin{equation}
		f^\text{\upshape ach}_x(x) = \exp\Big(- {\mu}_0 + {\mu}^2_1 \frac{x^2}{\bar{P}_\text{\upshape max}} \Big), \quad x\in [0, \sqrt{\bar{P}_\text{\upshape max}} ],
		\label{Eqn:PdfxAchievable}
	\end{equation} 
	\noindent\hspace*{-1pt}where $\mu_0 = \mu_1^2 + \ln\Big(\frac{ \sqrt{\bar{P}_\text{\upshape max}} }{1 + 2 \mu_1^2  \frac{\bar{P}_\text{\upshape req}}{\bar{P}_\text{\upshape max}}  }\Big)$. 
	Here, $\mu_1 \in \mathbb{R}_{+}$ is the solution of the following equation:
	\begin{equation}
		\frac{\bar{P}_\text{\upshape req}}{\bar{P}_\text{\upshape max}} = \frac{ \exp( \mu_1^2) }{ \sqrt{\pi} \mu_1 \text{\upshape Ei}( \mu_1 ) } - \frac{1}{ 2 \mu_1^2  }
		\label{Eqn:OptimalCoefMu3}
	\end{equation} 			
	\noindent with imaginary error function $\text{\upshape Ei}(\cdot)$.
	\label{Prop:OptimalSolution}
\end{proposition}
\begin{proof}
	Please refer to Appendix~\ref{Appendix:OptimalSolution}.
\end{proof}

Proposition~\ref{Prop:OptimalSolution} shows that if $\bar{P}_\text{\upshape req} \in [\frac{1}{3} \bar{P}_\text{\upshape max}, \bar{P}_\text{\upshape max})$, there is a tradeoff between the maximum achievable information rate and the average harvested power.
Furthermore, for given ${ \bar{A} }$ and $\bar{P}_\text{\upshape req}$, the pdf $f_x^\text{ach}$ and the maximum achievable information rate $J^*_{\bar{\mathcal{F}}_{ x }}$ depend on the power ratio $\frac{\bar{P}_\text{\upshape req}}{\bar{P}_\text{\upshape max}}$ in (\ref{Eqn:OptimalCoefMu3}).
The pdfs $f_x^\text{ach}$ obtained for different power ratios $\frac{\bar{P}_\text{req}}{\bar{P}_\text{max}}$ are shown in Fig.~\ref{Fig:PdfsX}.
In the following corollary, we provide the output pdf of $y$ in (\ref{Eqn:OutputPdf}) when the signal $x$ follows the pdf in (\ref{Eqn:PdfxAchievable}).

\begin{corollary}
	\label{Corollary:Optimal}
	If random variable $x$ follows the distribution with pdf $f^\text{\upshape ach}_x$ in (\ref{Eqn:PdfxAchievable}), the output pdf $f_y^\text{\upshape ach}$ can be expressed as follows:
	\begin{equation}
		\begin{aligned}
			f^\text{\upshape ach}_y(y; f^\text{\upshape ach}_x) = \frac{1}{\mu_4} \exp\Big( \frac{y^2}{\bar{P}_\text{\upshape max}} \frac{\mu_1^2}{\mu_4^2} - \mu_0\Big) \Big[ Q\Big(\frac{y - \mu_4^2 \sqrt{\bar{P}_\text{\upshape max}}}{\mu_4 \sigma}  \Big) - Q\Big(\frac{y}{\mu_4 \sigma}\Big) \Big]
		\end{aligned}
		\label{Eqn:OutputPdfAchievable}
	\end{equation}
	\noindent with $\mu_4 = \sqrt{1 - 2 \frac{\mu_1^2}{\bar{P}_\text{\upshape max}} \sigma^2}$.
\end{corollary}
\begin{proof}
	To obtain (\ref{Eqn:OutputPdfAchievable}), we express the output pdf $f_y^\text{ach}(y; f^\text{ach}_x)$ as follows:
	\begin{align}
		f_y^\text{ach}(y; f^\text{ach}_x) &= \int_{0}^{\sqrt{\bar{P}_\text{max}}} f^\text{ach}_x(x) f_n(y-x) \text{d}x \nonumber \\ 
		&= \int_{0}^{\sqrt{\bar{P}_\text{max}}} \frac{1}{\sqrt{2 \pi \sigma^2}} \exp \Big( - \mu_0 + \mu^2_1 \frac{x^2}{\bar{P}_\text{\upshape max}}-\frac{ (y-x)^2}{2 \sigma^2} \Big) \text{d}x \nonumber\\ 
		&= \mu_4^{-1} \exp\Big( \frac{y^2}{\bar{P}_\text{\upshape max}} \frac{\mu_1^2}{\mu_4^2} - \mu_0\Big) \Big[ Q\Big(\frac{y - \mu_4^2 \sqrt{\bar{P}_\text{\upshape max}}}{\mu_4 \sigma}  \Big) - Q\Big(\frac{y}{\mu_4 \sigma}\Big) \Big]
	\end{align}
	This concludes the proof.
\end{proof}

\begin{figure}[!t]
	\centering
	\includegraphics[draft = false, width=0.5\textwidth]{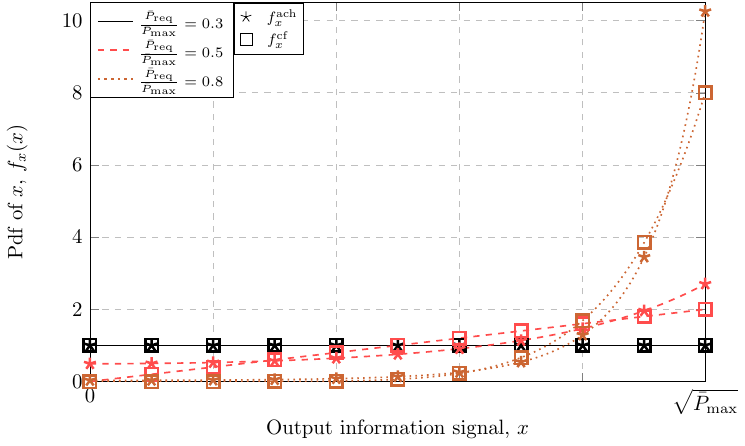}
	\caption{The pdfs $f_x^\text{ach}$ of $x$ in (\ref{Eqn:UniformDistributionX}) and (\ref{Eqn:PdfxAchievable}) and proposed closed-form pdfs $f_x^\text{cf}(x)$ in (\ref{Eqn:ClosedFormPdfX}) for different power ratios $\frac{\bar{P}_\text{req}}{\bar{P}_\text{max}}$.}
	\label{Fig:PdfsX}
\end{figure}

We note that if $\bar{P}_\text{\upshape req} \in [\frac{1}{3} \bar{P}_\text{\upshape max}, \bar{P}_\text{\upshape max})$, the maximum information rate $J^*_{\bar{\mathcal{F}}_{ x }} $, the pdf $f_x^\text{ach}$, and hence, $f_s^\text{ach}$ and mutual information $I_s(f^\text{\upshape ach}_s)$ require $\mu_1$ as solution of (\ref{Eqn:OptimalCoefMu3}), which, in general, can not be obtained in closed form.
{Therefore, in the following, we propose an iterative algorithm based on the bisection method for determining $\mu_1$.
The bisection approach is an iterative procedure that can be utilized to find a zero of the non-linear continuous function $g(\mu_1)-\frac{\bar{P}_\text{\upshape req}}{\bar{P}_\text{\upshape max}}$ inside a closed interval \cite{Burden1993}.
We note that function $g(\mu_1) = \frac{ \exp( \mu_1^2) }{ \sqrt{\pi} \mu_1 \text{\upshape Ei}( \mu_1 ) } - \frac{1}{ 2 \mu_1^2  }$ is monotonically increasing\footnotemark\hspace*{0pt} in its domain $\mathcal{D}\{g\} = \mathbb{R}_{+}$ and ${g}(\mu_1) \in [\frac{1}{3}, 1)$.
\footnotetext{The monotonicity of $g(\cdot)$ can be verified numerically by examining the first-order derivative $g'(\mu) = \frac{1}{\mu^3} + \frac{2 \exp(\mu^2)}{\sqrt{\pi} \text{Ei}(\mu)} \big[ 1 - \frac{1}{2\mu^2}  - \frac{\exp(\mu^2)}{\sqrt{\pi} \mu \text{Ei}(\mu)}] \geq 0, \forall \mu \in \mathbb{R}_{+}$.}}

First, we find a lower and an upper bound, $\mu_\text{lo}$ and $\mu_\text{up}$, of the initial interval for the bisection algorithm, such that $g(\mu_\text{lo}) \leq \frac{\bar{P}_\text{\upshape req}}{\bar{P}_\text{\upshape max}}$ and $g(\mu_\text{up}) \geq \frac{\bar{P}_\text{\upshape req}}{\bar{P}_\text{\upshape max}}$.
To this end, we create a lookup table comprising an array $\mathcal{M} = \{ \mu^1, \mu^2, \cdots \mu^{N_\mathcal{M}} \} \in \mathbb{R}_{+}^{N_\mathcal{M}}$, where $\mu^1 = 0$ and $\mu^i > \mu^j, \forall i > j$, of size $N_\mathcal{M}$ and a corresponding array $\mathcal{G} = \{g^1, g^2, \cdots, g^{N_\mathcal{M}} \}$ with $g^n = g(\mu^n), n \in \{1,2,\cdots, N_\mathcal{M}\}$.
We identify\footnotemark\hspace*{0pt} the array index $n_\text{up} = \min \{n \,\vert\, g^n \geq \frac{\bar{P}_\text{\upshape req}}{\bar{P}_\text{\upshape max}}\}$ and set $\mu_\text{up} = \mu^{n_\text{up}}$ and $\mu_\text{lo} = \mu^{n_\text{up} - 1}$.
\footnotetext{We note that $g(\mu) \to 1$ when $\mu \to \infty$. Therefore, if such $n_\text{up} \in \{1,2,\cdots, N_\mathcal{M}\}$ does not exist, we obtain the pdf of $x$ as $f_x^\text{ach} = \delta (x - \sqrt{\bar{P}_\text{max}})$. }

{Then, we utilize the bisection algorithm to find the unique parameter $\mu^*_1 \in [\mu_\text{lo}, \mu_\text{up}]$ that yields $g(\mu^*_1) = \frac{\bar{P}_\text{\upshape req}}{\bar{P}_\text{\upshape max}}$.
In particular, in iteration $p \geq 1$, we calculate the mid point of the interval, $\mu^\text{c} = \frac{\mu_\text{lo} + \mu_\text{up}}{2}$.
We compare the value of function $g(\cdot)$ at the mid point with power ratio $\frac{\bar{P}_\text{\upshape req}}{\bar{P}_\text{\upshape max}}$ and we set $\mu_\text{up} = \mu^\text{c} $ if $g(\mu^\text{c}) > \frac{\bar{P}_\text{\upshape req}}{\bar{P}_\text{\upshape max}}$ and $\mu_\text{lo} = \mu^\text{c} $, otherwise.
Finally, if $| \mu_\text{up} - \mu_\text{lo}| \leq \epsilon$, where $\epsilon$ is a predefined desired accuracy, we set $\mu_1 = \frac{\mu_\text{up} + \mu_\text{lo}}{2}$ and stop the bisection search.
We note that the obtained value satisfies $\vert \mu_1 - \mu_1^* \vert \leq \epsilon$, and thus, the bisection algorithm converges to the unique solution $\mu^*_1$ of the non-linear equation $g(\mu^*_1) = \frac{\bar{P}_\text{\upshape req}}{\bar{P}_\text{\upshape max}}$.}
With the determined parameter $\mu_1$, we find $\mu_0$ and $f_x^\text{ach}$ in Proposition~\ref{Prop:OptimalSolution}.
The corresponding input pdf is given by $f_s^\text{ach}(s) = \frac{\partial}{\partial s} F_x^\text{ach}(\hat{\psi} (|hs|^2))$, where $F_x^\text{ach}(x) = \int_{0}^x f_x^\text{ach}(\tilde{x}) \text{d}\tilde{x} $ is the cdf of $x$.
Similar to $f_s^\text{opt}$, pdf $f_s^\text{ach}$ exists and is unique due to the monotonicity of the equivalent EH model $\hat{\psi}(\cdot)$.

\begin{algorithm}[!t]	
	\small				
	\SetAlgoNoLine%
	Initialize: Peak amplitude $\bar{A}$, required average harvested power $\bar{P}_\text{\upshape req}$, EH model $\psi(\cdot)$, channel coefficient $h$, lookup table arrays $\mathcal{M}$ and $\mathcal{G}$, and accuracy $\epsilon$. \\	
	1. Find $\bar{P}_\text{\upshape max} = \max_{\rho \in [0, |h\bar{A}|^2 ] } \hat{\psi}(\rho)$\\
	\uIf{$\bar{P}_\text{\upshape req} \in [0, \frac{1}{3} \bar{P}_\text{\upshape max}]$}
	{2. Set pdf $f_x^\text{ach}(x) = \frac{1}{\sqrt{\bar{P}_\text{\upshape max}}}$ with $x \in [0, \sqrt{\bar{P}_\text{\upshape max}}]$\\}
	\uElse
	{3. Identify $n_\text{up} = \min \{n \,\vert\, g^n \geq \frac{\bar{P}_\text{\upshape req}}{\bar{P}_\text{\upshape max}}\}$, where $g^n \in \mathcal{G}, \forall n$, and set $\mu_\text{up} = \mu^{n_\text{up}} \in \mathcal{M}$ and $\mu_\text{lo} = \mu^{n_\text{up} -1} \in \mathcal{M}$\\
	4. Utilize the bisection approach to find $\mu_1$:\\
	\While{$| \mu_\text{\upshape up} - \mu_\text{\upshape lo} | > \epsilon$}{
		4.1. Find $\mu^{c} = \frac{\mu_\text{up} + \mu_\text{lo}}{2}$\\
		4.2. If $g(\mu^\text{c}) > \frac{\bar{P}_\text{\upshape req}}{\bar{P}_\text{\upshape max}}$ set $\mu_\text{up} = \mu^{c}$, else set $\mu_\text{lo} = \mu^{c}$\\
	}
	 5. Set $\mu_1 = \frac{\mu_\text{up} + \mu_\text{lo}}{2}$\\
	 6. Determine the pdf $f_x^\text{ach}(x) = \exp(- {\mu}_0 + {\mu}_1^2 \frac{x^2}{\bar{P}_\text{\upshape max}})$, where $x\in [0, \sqrt{\bar{P}_\text{\upshape max}} ]$ and $\mu_0 \in \mathbb{R}_{+}$ is given in Proposition~\ref{Prop:OptimalSolution} \\
	}
	7. Find cdf $F_x^\text{ach}(x) = \int_{0}^{x} f_x^\text{ach}(\tilde{x}) \text{d}\tilde{x}$ and input pdf $f_s^\text{ach}(s) = \frac{\partial}{\partial s} F_x^\text{ach}(\hat{\psi} (|hs|^2))$\\
	\textbf{Output:} Input pdf $f_s^\text{ach}$.
	\caption{\strut Algorithm for determining the input pdf $f_s^\text{ach}$. }
	\label{Alg:AchievableSolution}
\end{algorithm}

The proposed iterative algorithm for determining the input pdf $f_s^\text{ach}$ as a suboptimal solution of (\ref{Eqn:GeneralOptimizationProblem}) is summarized in \textbf{Algorithm~\ref{Alg:AchievableSolution}}.
We note that for $\bar{P}_\text{\upshape req} \in [0, \frac{1}{3} \bar{P}_\text{\upshape max}]$, the input pdf $f_s^\text{ach}$ can be obtained in closed form.
However, if $\bar{P}_\text{\upshape req} \in (\frac{1}{3} \bar{P}_\text{\upshape max}, \bar{P}_\text{\upshape max})$, the proposed algorithm still requires a numerical solution of (\ref{Eqn:OptimalCoefMu3}).
{Thus, the computational complexity of \textbf{Algorithm~\ref{Alg:AchievableSolution}} does not depend on accuracy $\epsilon$ if $\bar{P}_\text{\upshape req} \in [0, \frac{1}{3} \bar{P}_\text{\upshape max}]$, and is given by  $\mathcal{O}( \log_2{\epsilon}^{-1})$, otherwise \cite{Burden1993}.
Since, for high required harvested powers, the resulting computational complexity of the proposed suboptimal solution may still be too high for practical implementation, in the following, we propose a closed-form suboptimal pdf $f_x^\text{cf}$ as an alternative suboptimal solution of (\ref{Eqn:RefOutputOptimizationProblem}).}
\subsection{Proposed Closed-form Input Distribution}
\label{Section:ProposedDistribution}

We note that when the required average harvested power is low, i.e., $\bar{P}_\text{req} \leq \frac{1}{3} \bar{P}_\text{max}$, the pdf $f^\text{ach}_x \in \bar{\mathcal{F}}_x$ maximizing the achievable information rate in (\ref{Eqn:MI_LowerBound}) is uniform and given by (\ref{Eqn:UniformDistributionX}).
Furthermore, for high required average harvested powers satisfying $\bar{P}_\text{req} \in (\frac{1}{3} \bar{P}_\text{max}, \bar{P}_\text{max})$, the pdf $f^\text{ach}_x (x)$ in (\ref{Eqn:PdfxAchievable}) is a monotonic increasing function of $x \in [0, \sqrt{\bar{P}_\text{max}}]$.
Finally, if $\bar{P}_\text{req} = \bar{P}_\text{max}$, the optimal pdf solving (\ref{Eqn:RefOutputOptimizationProblem}) is trivial and is given by $f_x^\text{opt}(x) = \delta(x - \sqrt{\bar{P}_\text{max}})$.
Motivated by these observations, we propose to employ the following closed-form suboptimal pdf $f_x^\text{cf}$ of $x$, which is uniform, monotonically increasing, and reduces to a Dirac delta function for $\bar{P}_\text{req} \leq \frac{1}{3} \bar{P}_\text{max}$, $\bar{P}_\text{max} < \bar{P}_\text{req} < \frac{1}{3} \bar{P}_\text{max}$, and $\bar{P}_\text{req} = \bar{P}_\text{max}$, respectively:
\begin{equation}
	f_x^\text{cf}(x) = \alpha P_\text{max}^{-\frac{\alpha}{2}} x^{\alpha-1}, \quad x \in [0, \sqrt{\bar{P}_\text{max}}].
	\label{Eqn:ClosedFormPdfX}
\end{equation}

We note that $f_x^\text{cf}$ is the pdf of the scaled Beta distribution (or, equivalently, Kumaraswamy distribution) with parameters $\alpha$ and $1$ \cite{Cover2012}.
It can be shown that for any $\alpha > 0$, $f_x^\text{cf}$ with $\mathcal{D}\{f_x^\text{cf}\} = [0, \sqrt{\bar{P}_\text{max}}]$ is a valid pdf, i.e., we have $\int f_x^\text{cf}(x) \text{d}x = 1$.
Furthermore, the proposed pdf $f^\text{cf}_x \in \bar{\mathcal{F}}_x$, i.e., $f_x^\text{cf}$ satisfies constraint (\ref{Eqn:RefOutputOptConstr1}), if and only if $\alpha \geq \alpha_{\bar{\mathcal{F}}_x} \triangleq \frac{2 \bar{P}_\text{req} }{ \bar{P}_\text{max} - \bar{P}_\text{req}}$.
Moreover, if $\frac{\bar{P}_\text{req}}{\bar{P}_\text{max}} = \frac{1}{3}$, we have $\alpha = \alpha_{\bar{\mathcal{F}}_x} = 1$, and similar to $f_x^\text{ach}$ in (\ref{Eqn:UniformDistributionX}), pdf $f_x^\text{cf}(x) = P_\text{max}^{-\frac{1}{2}}$ is uniform.
Thus, we set parameter $\alpha = \max \{\alpha_{\bar{\mathcal{F}}}, 1\}$.
The proposed pdf $f_x^\text{cf}$ is shown in Fig.~\ref{Fig:PdfsX} for different power ratios $\frac{\bar{P}_\text{req}}{\bar{P}_\text{max}}$.
Finally, we obtain the corresponding transmit signal pdf $f_s^\text{cf}$ in closed form as $f_s^\text{cf}(s) = \frac{\partial}{\partial s} F_x^\text{cf}(\hat{\psi} (|hs|^2))$, where $F_x^\text{cf}(x) = \int_{0}^x f_x^\text{cf}(\tilde{x}) \text{d}\tilde{x} $ is the cdf of $x$.

In the following proposition, we derive the achievable information rate $J_x(\cdot)$ in (\ref{Eqn:MI_LowerBound}) when the pdf of $x$ is given by (\ref{Eqn:ClosedFormPdfX}).
\begin{proposition}
	\label{Prop:ProposedMutualInformation}
	If random variable $x$ follows pdf $f_x^\text{\upshape cf}$ in (\ref{Eqn:ClosedFormPdfX}), the achievable mutual information in (\ref{Eqn:MI_LowerBound}) is given by
	\begin{align}
		J_x(f_x^\text{\upshape cf}) = \frac{1}{2} \ln\Big(1 + \frac{  \bar{P}_\text{\upshape max} e^{2 \frac{\alpha-1}{\alpha}}}{2\pi e \sigma^2 \alpha^2 }\Big).
		\label{Eqn:CF_AchievableRate}
	\end{align}
\end{proposition}
\begin{proof}
	First, we express the entropy of random variable $x$, which follows pdf $f_x^\text{cf}(x)$ in (\ref{Eqn:ClosedFormPdfX}), as follows:
	\begin{align}
		h_x(f_x^\text{cf}) &
		= -\int_0^{\sqrt{\bar{P}_\text{max}}} \alpha \bar{P}_\text{max}^{-\frac{\alpha}{2}} x^{\alpha-1} \ln (\alpha \bar{P}_\text{max}^{-\frac{\alpha}{2}} x^{\alpha-1}) \, \text{d}x \nonumber \\
		&=-\int_0^{1} (\alpha-1) \ln y \,\text{d}y^\alpha  -\int_0^{1} \ln (\alpha \bar{P}_\text{max}^{-\frac{1}{2}} ) \, \text{d}y^\alpha = \frac{\alpha-1}{\alpha} - \ln\frac{\alpha}{\sqrt{\bar{P}_\text{max}}}.
		\label{Eqn:Col2Eqn1}
	\end{align}
	Then, to obtain (\ref{Eqn:CF_AchievableRate}), we substitute (\ref{Eqn:Col2Eqn1}) into the definition of achievable rate $J_x(f_x^\text{cf})$ in Lemma~\ref{Lemma:EPI}.
	This concludes the proof.
\end{proof}

Proposition~\ref{Prop:ProposedMutualInformation} reveals that for $\alpha = 1$, i.e., for $\bar{P}_\text{req} \leq \frac{1}{3} \bar{P}_\text{max}$, the pdf $f_x^\text{cf}$, which is defined in (\ref{Eqn:ClosedFormPdfX}) and shown in Fig.~\ref{Fig:PdfsX}, yields the maximum achievable information rate $J_x(f_x^\text{\upshape cf}) = J^*_{\bar{\mathcal{F}}_{ x }}$ in Proposition~\ref{Prop:UniformDistribution}.
However, if $\bar{P}_\text{req} \in ( \frac{1}{3} \bar{P}_\text{max},  \bar{P}_\text{max})$, the achievable information rate in (\ref{Eqn:CF_AchievableRate}) may be lower than the maximum value in (\ref{Eqn:OptMI}), i.e., we have $J_x(f_x^\text{\upshape cf}) \leq J^*_{\bar{\mathcal{F}}_{ x }}$.
In the next section, we numerically evaluate the performance of THz SWIPT systems and show that all obtained input pdfs, i.e., $f_s^\text{opt}$, $f_s^\text{ach}$, and $f_s^\text{cf}$, yield practically identical mutual information $I_s(\cdot)$ for all considered values of $A$ and $\bar{P}_\text{req}$.


\section{Numerical Results}
\label{Section:NumericalResults}
In this section, we evaluate the performance of the considered THz SWIPT system via numerical simulations.
To this end, we first specify the parameters of the adopted simulation setup and tune the parameters of the proposed parametric EH model to match circuit simulation results.
Next, we evaluate the mutual information and achievable information rate of the THz SWIPT system when the derived input signal distributions are adopted at the TX.
Finally, we study the tradeoff between the mutual information and average harvested power at the RX.
\subsection{Simulation Setup}
\label{Section:SimulationSetup}

{For our numerical simulations, we model the channel between TX and RX as follows \cite{Serghiou2022, Rong2017, Pan2022}:
	\begin{equation}
		h = h_\text{spr} h_\text{abs} h_\text{mis} h_\text{f},
		\label{Eqn:THzChannelModel}
	\end{equation}   
\noindent where $h_\text{spr}, h_\text{abs}, h_\text{mis}, h_\text{f} \in \mathbb{R}$ are the free space spreading, molecular absorption, antenna misalignment, and multipath fading loss coefficients, respectively.
Here, we model the free space spreading loss $h_\text{spr}$ as follows \cite{Zhu2022, Pan2022, Serghiou2022}
\begin{equation}
	h_\text{spr} = \frac{c_l}{4 \pi f_\text{c} d} \sqrt{G_\text{T} G_\text{R}},
	\label{Eqn:SpreadingLoss}
\end{equation}   
\noindent where $c_l$ denotes the speed of light, $f_\text{c} = \SI{300}{\giga\hertz}$ is the carrier frequency\footnote{For our simulations, we utilize $f_\text{c}=\SI{300}{\giga\hertz}$ since the I-V characteristics of the Keysight ADS \cite{ADS2017} RTD design from \cite{Clochiatti2022} are shown to accurately match the measurements results in \cite{Clochiatti2022} for this frequency.}, and $d = \SI{0.1}{\meter}$ is the distance\footnote{For the efficient compact THz antenna designs in \cite{Fan2017}, the adopted distance between the TX and RX exceeds the Fraunhofer distance, and thus, the THz RX operates in the far-field regime \cite{Mayer2023}. We note that the analysis of THz SWIPT systems operating in the near-field is an interesting direction for future work.} between TX and RX.
The spreading loss model in (\ref{Eqn:SpreadingLoss}) has been shown to accurately match experimental results for short-distance sub-THz communication networks \cite{Khalid2019, Papasotiriou2021}.
Furthermore, we set the TX and RX antenna gains to $G_\text{T} = \SI{25}{\dBi}$ and $G_\text{R} = \SI{15}{\dBi}$, respectively \cite{Priebe2011, Fan2017}.
We set $h_\text{mis} = 0.95$ \cite{Boulogeorgos2022, Lee2019} and model the molecular absorption channel coefficient as $h_\text{abs} = e^{-\frac{1}{2} \kappa d}$, where $\kappa = 3\cdot10^{-3}\SI{}{m^{-1}}$ is a molecular absorption coefficient at carrier frequency $f_\text{c} = \SI{300}{\giga\hertz}$ \cite{Sarieddeen2021,  Zhu2022, Boulogeorgos2022}.
Finally, unless specified otherwise, we model the multipath fading coefficient ${h}_\text{f}$ as a Rician distributed random variable with Rician factor $1$ \cite{Serghiou2022, Boulogeorgos2019}.
}

In Algorithm~\ref{Alg:OptimalSolution}, we set the grid size to $K = 1000$, while the accuracy and the lookup table arrays in Algorithm~\ref{Alg:AchievableSolution} are set to $\epsilon = 10^{-3}$ and $\mathcal{M} = \{0, 10^{-1}, 10^{0}, 10^{1}, 10^{2}, 10^{3}, 10^{4}\}$ and $\mathcal{G} = \{\tilde{g} \,\vert\, \tilde{g} = g(\mu), \forall \mu \in \mathcal{M} \}$, respectively. 
The noise variance at the output of the RX is set to $\sigma^2 = \SI{-50}{\dBm}$.

\subsection{RX Circuit Simulations and EH Model Fitting}
\label{Section:ModelValidation}

\begin{table*}[!t]
	\centering
	\footnotesize
	\caption{Tuned parameters of the EH model in (\ref{Eqn:EHmodel}).}
	\hspace*{-10pt}
	\begin{tabular}{|m{0.12\textwidth} | m{0.12\textwidth} || m{0.12\textwidth} | m{0.12\textwidth} || m{0.12\textwidth} |m{0.12\textwidth} | m{0.13\textwidth}|}
		\hline 
		\multicolumn{2}{|c||}{RTD design in \cite{Clochiatti2022}} & \multicolumn{2}{c||}{RTD design with improved $I_\text{rev}$} & \multicolumn{3}{c|}{RTD design with improved $U_\text{br}$}\\
		\hline
		\multicolumn{2}{|c||}{$\rho_1 = \SI{1.8}{\milli\watt}, \rho_2 = \SI{2.4}{\milli\watt}$}& 
		\multicolumn{2}{c||}{$\rho_1 = \SI{2.1}{\milli\watt}, \rho_2= \SI{3}{\milli\watt}$}& \multicolumn{3}{c|}{$\rho_1 = \SI{4.1}{\milli\watt}, \rho_2 = \SI{4.17}{\milli\watt}, \rho_3 = \SI{6.18}{\milli\watt}$}\\
		\hline
		\multicolumn{1}{|c|}{$\varphi_1(\cdot)$} & 
		\multicolumn{1}{c||}{$\varphi_2(\cdot)$} & \multicolumn{1}{c|}{$\varphi_1(\cdot)$} & 
		\multicolumn{1}{c||}{$\varphi_2(\cdot)$} & \multicolumn{1}{c|}{$\varphi_1(\cdot)$} & \multicolumn{1}{c|}{$\varphi_2(\cdot)$} & \multicolumn{1}{c|}{$\varphi_3(\cdot)$} \\
		\hline
		\makecell[l]{$B_1 = \SI{71.6}{\micro\watt}$\\ $\alpha_1 = 1.432$ \\ $\beta_1 = 0.778$ \\ $\theta_1 = 2174.9$} 
		& \makecell[l]{$B_2 = \SI{25}{\micro\watt}$\\ $\alpha_2 = 1.841$ \\ $\beta_2 = 0.445$ \\ $\theta_2 = 956.8$}
		& \makecell[l]{$B_1 = \SI{315}{\micro\watt}$\\ $\alpha_1 = 1.46$ \\ $\beta_1 = 0.527$ \\ $\theta_1 = 3580$} 
		& \makecell[l]{$B_2 = \SI{104}{\micro\watt}$\\ $\alpha_2 = 2.601$ \\ $\beta_2 = 0.703$ \\ $\theta_2 = 1100$}
		& \makecell[l]{$B_1 = \SI{3.6}{\milli\watt}$\\ $\alpha_1 = 1.534$ \\ $\beta_1 = 0.289$ \\ $\theta_1 = 241.6$}
		& \makecell[l]{$B_2 = \SI{535}{\micro\watt}$\\ $\alpha_2 = 3.492$ \\ $\beta_2 = 10^{4}$ \\ $\theta_2 = 1692$}
		& \makecell[l]{$B_3 = \SI{2.85}{\milli\watt}$\\ $\alpha_3 = 1.492$ \\ $\beta_3 = 0.244$ \\ $\theta_3 = 294.8$}\\
		\hline
	\end{tabular}
	\label{Table_EhModel}
\end{table*}
\begin{figure}[!t]
	\centering
	\includegraphics[draft = false, width=0.48\textwidth]{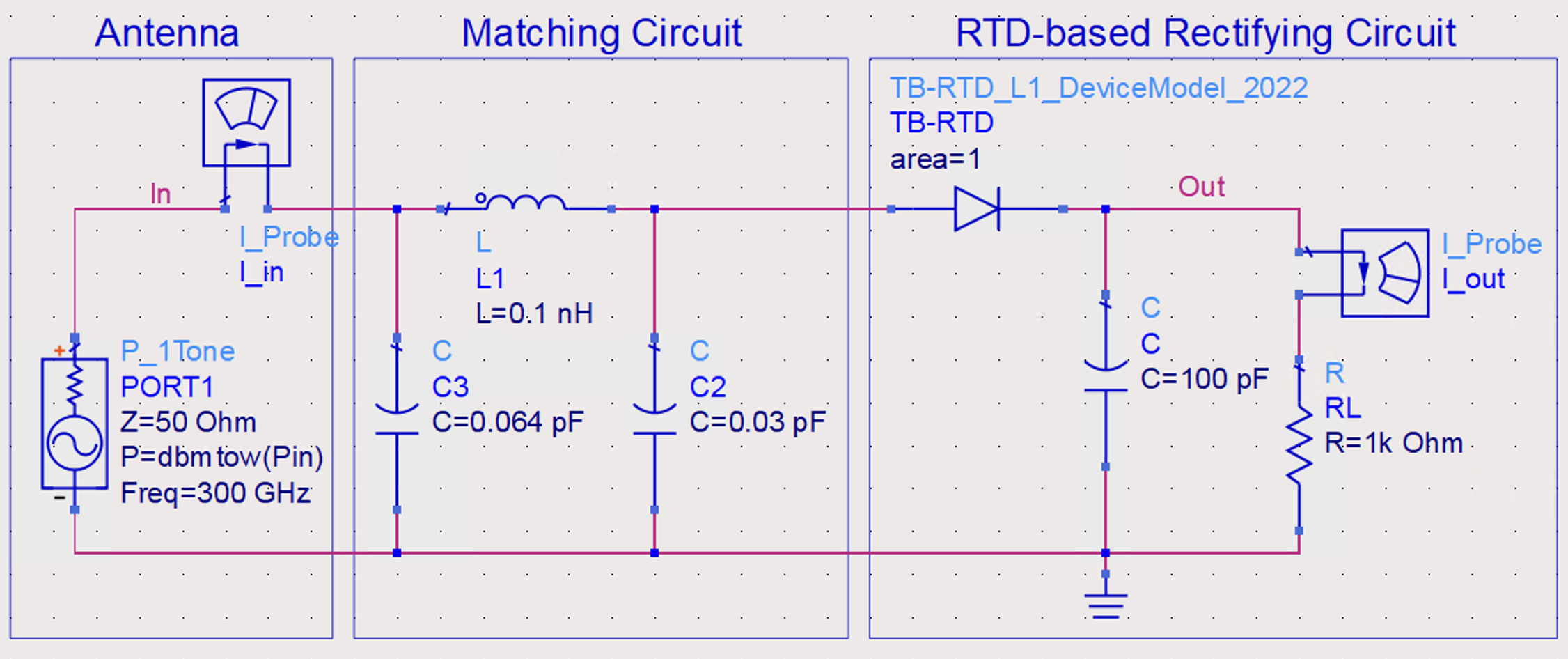}
	\caption{Circuit simulation setup in Keysight ADS \cite{ADS2017}.}
	\label{Fig:ADSModel}\vspace*{-5pt}
\end{figure}

First, we tune the parameters of the proposed general EH model $\psi(\cdot)$ to match the circuit simulation results.
To this end, we model the RTD-based EH circuit shown in Fig.~\ref{Fig:SystemModel} with circuit parameters $C = \SI{100}{\pico\farad}$ and $R_\text{L} = \SI{1}{\kilo\ohm}$ \cite{Morsi2019} using harmonic balance Keysight ADS \cite{ADS2017} circuit simulations, as shown in Fig.~\ref{Fig:ADSModel}.
In particular, we assume an antenna impedance of $\SI{50}{\ohm}$, and for each considered input power $\rho$, we design the matching circuit to achieve near-perfect impedance matching between the antenna and the rectifier of the EH circuit, see Fig.~\ref{Fig:ADSModel} \cite{Morsi2019, Shanin2020}.
To investigate the impact of the breakdown voltage and leakage current of the RTD on the performance of THz SWIPT systems, in our simulations, we consider not only the Keysight ADS \cite{ADS2017} design of the triple-barrier RTD proposed in \cite{Clochiatti2022}, but also two improved RTD designs.
In particular, we adopt RTD designs\footnote{The developed Keysight ADS RTD designs and EH models can be downloaded from https://gitlab.com/NikitaShaninFAU/rtd-designs-for-thz-eh.} with lower reverse leakage current $I_\text{rev}$ and higher breakdown voltage $U_\text{br}$ compared to the design in \cite{Clochiatti2022}, respectively, whose I-V characteristics are shown in Fig.~\ref{Fig:IV_Curve}.

{In Fig.~\ref{Fig:MatchedEhModel}, for all considered RTDs, we show the proposed EH model $\psi(\cdot)$ tuned to match the circuit simulation results, which are obtained using the Keysight ADS circuit simulation tool \cite{ADS2017}.
The respective model parameters are summarized in Table~\ref{Table_EhModel}.
Furthermore, in Fig.~\ref{Fig:MatchedEhModel}, we depict the linear EH model in \cite{Zhang2013} referred to as Baseline EH Model 1.
For this model, the EH conversion efficiency is set to match the ADS circuit simulation results for the RTD design in \cite{Clochiatti2022} assuming $\rho = \rho_\text{max}$, i.e., Baseline EH Model 1 passes through the points $(0, \psi(0))$ and $(\rho_\text{max}, \psi(\rho_\text{max}))$.
Finally, to reveal the importance of an accurate modelling of RTD-based EH circuits in the THz band, we also show the non-linear circuit-based EH model that was proposed in \cite{Morsi2019} to model the non-linear behavior of EH circuits based on Schottky diodes in the GHz band.
We refer to this circuit-based EH model as Baseline EH Model 2 in Fig.~\ref{Fig:MatchedEhModel}.}

{ We observe from Fig.~\ref{Fig:MatchedEhModel} that the proposed EH models precisely match the Keysight ADS \cite{ADS2017} simulation results and accurately capture the non-linear and non-monotonic behavior of the considered EH circuits employing RTDs in the THz band.
In contrast, although the non-linear circuit-based EH model in \cite{Morsi2019} can capture the non-linearity of EH circuits employing Schottky diodes in the GHz band, both baseline schemes are not able to accurately characterize the non-monotonic behavior of RTD-based RXs since the electrical characteristics of Schottky diodes and RTDs are significantly different.}
Next, we note from Fig.~\ref{Fig:MatchedEhModel} that for all considered RTDs, as expected, the instantaneous harvested power grows with the input signal power for $\rho \in [0, \rho_1)$.
For the RTD design in \cite{Clochiatti2022} and the RTD with improved $I_\text{rev}$, the power harvested at the EH receiver decreases with the input power in the corresponding intervals $[\rho_1, \rho_\text{max}]$ indicated in Fig.~\ref{Fig:MatchedEhModel} until the maximum power level $\rho = \rho_\text{max}$ is reached.
Since the maximum power level $\rho_\text{max}$ depends not only on the breakdown voltage $U_\text{br}$, but also on the reverse-bias current flow $I_\text{d}$ for $U_\text{d} \in (-U_\text{br}, 0)$, the values of $\rho_\text{max}$ for both diodes are similar but not identical \cite{Tietze2012}.
Furthermore, since for both RTD designs, the breakdown voltage $U_\text{br}$ is low, functions $\varphi_1(\cdot)$ and $\varphi_2(\cdot)$ are sufficient for EH modelling, and thus, we have $\rho_2 = \rho_\text{max}$.
{On the contrary, we observe from Fig.~\ref{Fig:MatchedEhModel} that for the RTD with higher breakdown voltage $U_\text{br}$, the harvested power at the EH receiver decreases for $\rho \in [\rho_1, \rho_2]$ and increases for $\rho \in [\rho_2, \rho_\text{max}]$.
In other words, compared to the EH circuit equipped with the RTD from \cite{Clochiatti2022}, the corresponding EH circuit can operate at a higher received signal power level exceeding $\SI{5}{\milli\watt}$ and achieves a maximum instantaneous harvested power of around $\SI{1}{\milli\watt}$.}
Therefore, for modelling the corresponding EH circuit, not only functions $\varphi_1(\cdot)$ and $\varphi_2(\cdot)$ are required, but also function $\varphi_3(\cdot)$.
Finally, we observe from Fig.~\ref{Fig:MatchedEhModel} that the lower reverse leakage current and the higher breakdown voltage of the improved diodes yield higher EH efficiencies, and thus, are preferable for SWIPT, when the received power $\rho$ is low and high, respectively.
In the following, we will utilize the obtained EH models to design the input signal distributions and evaluate the performance of THz SWIPT systems.

\begin{figure}[!t]
	\centering
	\includegraphics[draft = false, width=0.49\textwidth]{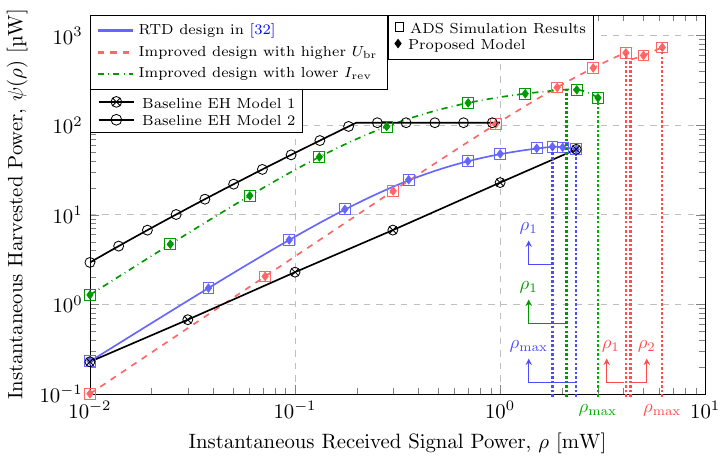}
	\caption{Proposed EH model tuned to match the circuit simulation results.}
	\label{Fig:MatchedEhModel}\vspace*{-5pt}
\end{figure}

\subsection{Input Distribution, Mutual Information, and Achievable Rate}
\label{Section:SimulationsMutInf}
\begin{figure*}[!t]
	\centering
	\subfigure[High peak TX amplitude $A = \SI{2}{\volt}$]{
		\includegraphics[draft=false, width=0.45\textwidth]{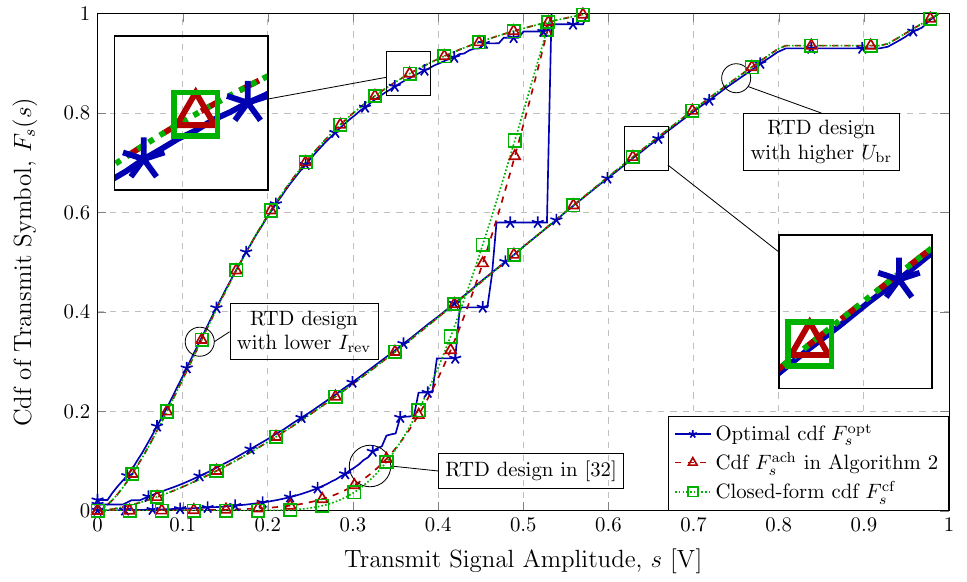} \label{Fig:Results_Distr_HighA}}\hspace*{27pt}
	\subfigure[Low peak TX amplitude $A = \SI{0.45}{\volt}$]{
		\includegraphics[draft=false, width=0.45\textwidth]{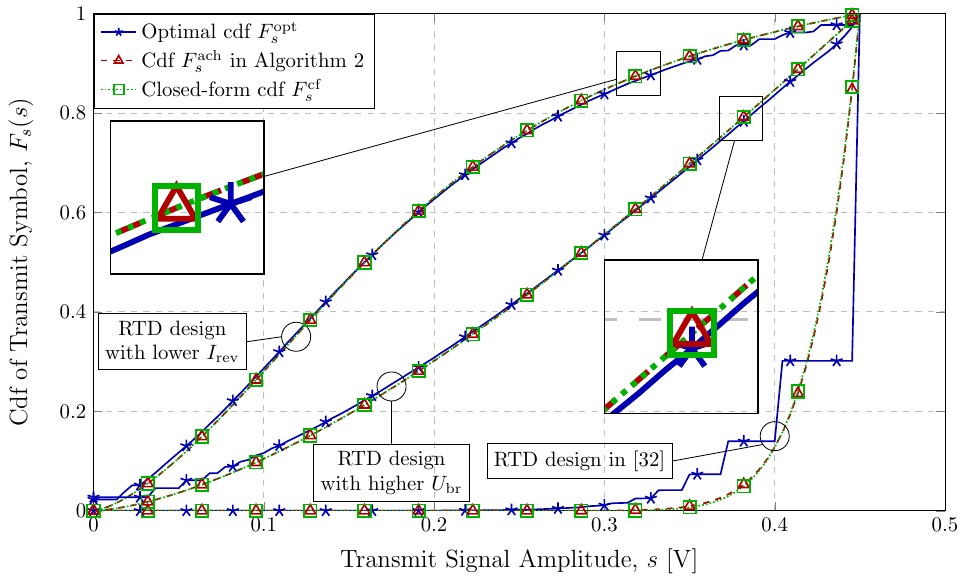} \label{Fig:Results_Distr_LowA}}
	\caption{Cdfs $F_s^\text{opt}, F_s^\text{ach}$, and $F_s^\text{cf}$ obtained for different RTD designs and peak TX amplitudes $A$.}
	\label{Fig:ResultsDistributions}\vspace*{-5pt}
\end{figure*}

In this section, we numerically investigate the performance of THz SWIPT systems when the proposed input distributions are adopted at the TX.
First, for the considered RTD designs and different peak TX amplitudes $A$, we determine and plot the input cdfs $F_s^\text{opt}$ and $F_s^\text{ach}$ obtained with Algorithms~\ref{Alg:OptimalSolution} and \ref{Alg:AchievableSolution}, respectively, in Fig.~\ref{Fig:ResultsDistributions}.
Furthermore, in Fig.~\ref{Fig:ResultsDistributions}, we plot the input cdfs $F_s^\text{cf}$ obtained for the proposed closed-form pdf $f_x^\text{cf}$ in (\ref{Eqn:ClosedFormPdfX}).
{For the results shown in Fig.~\ref{Fig:ResultsDistributions}, we assume a line-of-sight channel between TX and RX without multipath propagation, set the multipath fading coefficient to one, i.e., ${h}_\text{f} = 1$, and adopt a required average harvested power of $\bar{P}_\text{req} = \SI{50}{\micro\watt}$, which is smaller than $\bar{P}_\text{max}$ for all considered RTD designs.}
For the cdfs shown in Figs.~\ref{Fig:Results_Distr_HighA} and \ref{Fig:Results_Distr_LowA}, we adopt high and low peak transmit amplitudes of $A = \SI{2}{\volt}$ and $A = \SI{0.45}{\volt} \leq \frac{\sqrt{\rho_\text{max}}}{|h|}$, respectively.

{First, we observe from Fig.~\ref{Fig:ResultsDistributions} that the suboptimal cdfs $F_s^\text{ach}$, which maximize the achievable rates $J^*_{\bar{\mathcal{F}}_x}$, and the proposed closed-form cdfs $F_s^\text{cf}$ are almost identical for both considered RTD designs and both considered values of $A$.
Next, we observe that in contrast to the suboptimal input distributions corresponding to $F_s^\text{ach}$ and $F_s^\text{cf}$, the optimal input distributions are characterized by step functions $F_s^\text{opt}$, and hence, the optimal transmit signal alphabet is discrete \cite{Varshney2008, Morsi2019}.}
Moreover, we note that if the peak transmit amplitude $A$ is high, the maximum transmit signal amplitude $\bar{A}$ is determined by the maximum harvested power $\rho_\text{max}$, which is lower for the RTD developed in \cite{Clochiatti2022} compared to the improved RTD designs, as shown in Fig.~\ref{Fig:MatchedEhModel}.
Finally, we observe that for the lower peak transmit amplitude $A = \SI{0.45}{\volt}$, significantly higher transmit signal amplitudes are required for the RTD design from \cite{Clochiatti2022} compared to the improved RTDs.
This is due to the lower EH efficiency and the lower maximum harvested power of the RX circuit equipped with the RTD from \cite{Clochiatti2022}, see Fig.~\ref{Fig:MatchedEhModel}.

In Fig.~\ref{Fig:ResultsMutInf}, we plot the average mutual information obtained for different RTD designs and different required average harvested powers $\bar{P}_\text{req}$.
{For the results in Fig.~\ref{Fig:ResultsMutInf}, we assume that the channel coefficient $h$ is constant during the transmission of each codeword and perfectly known at the TX and RX, model the multipath fading coefficient $h_\text{f}$ as a Rician distributed random variable with unit Rician factor, and average all obtained simulation results over 1000 \gls{iid} random realizations of $h_\text{f}$.
In particular, for each random channel realization of $h$, we first determine pdfs $f_s^\text{opt}$, $f_s^\text{ach}$, and $f_s^\text{cf}$ with Algorithm~\ref{Alg:OptimalSolution}, Algorithm~\ref{Alg:AchievableSolution}, and (\ref{Eqn:ClosedFormPdfX}), respectively.
Next, for the obtained pdfs, we calculate the mutual information $I_s(f_s^\text{opt})$, $I_s(f_s^\text{ach})$, and $I_s(f_s^\text{cf})$ and achievable information rates ${J}_x(f_x^\text{ach})$ and ${J}_x(f_x^\text{cf})$.
Finally, we determine the average mutual information and average achievable rates at the RX and plot them in Fig.~\ref{Fig:ResultsMutInf}.}
For the results in Figs.~\ref{Fig:ResultsMutInf_Realistic}, \ref{Fig:ResultsMutInf_ImprovedId}, and \ref{Fig:ResultsMutInf_ImprovedVb}, we consider the cases, where the RX circuit is equipped with the RTD design proposed in \cite{Clochiatti2022}, the RTD design with improved $I_\text{rev}$, and the RTD design with improved $U_\text{br}$, respectively.
Furthermore, for each considered RTD design, we study the performance of a pure wireless information transfer system with $\bar{P}_\text{req} = 0$ and a THz SWIPT system with low and high required average harvested powers of $\bar{P}_\text{req} = 0.4 \hat{P}_\text{max}$ and $\bar{P}_\text{req} = 0.8 \hat{P}_\text{max}$, respectively, where $\hat{P}_\text{max} = \max_{\rho \in [0, \rho_\text{max} ] } \psi(\rho)$ is the maximum harvested power at the RX.

\begin{figure*}[!t]
	\centering
	\subfigure[RTD design from \cite{Clochiatti2022}]{
		\includegraphics[draft=false, width=0.48\textwidth]{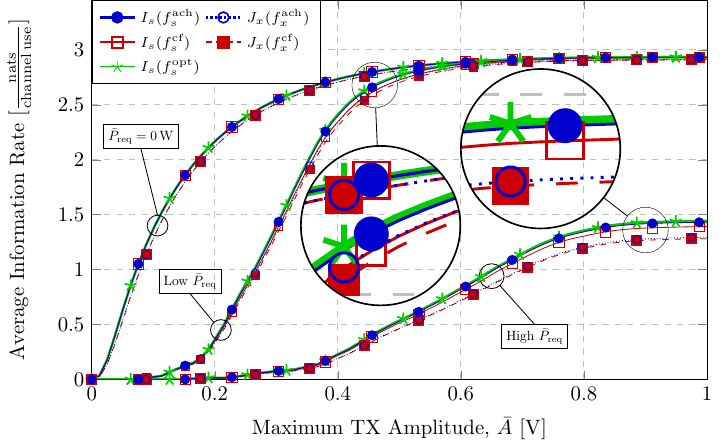} \label{Fig:ResultsMutInf_Realistic}}
	\subfigure[Improved RTD design with lower $I_\text{rev}$]{
		\includegraphics[draft=false, width=0.48\textwidth]{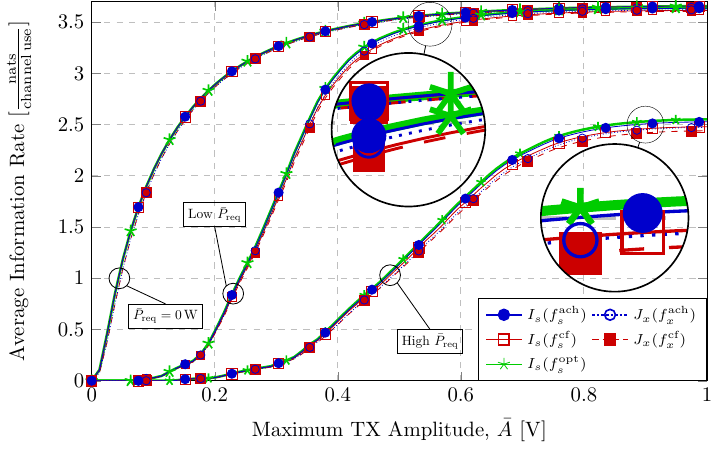} \label{Fig:ResultsMutInf_ImprovedId}}\hspace*{5pt}
	\subfigure[Improved RTD design with higher $U_\text{br}$]{
		\includegraphics[draft=false, width=0.48\textwidth]{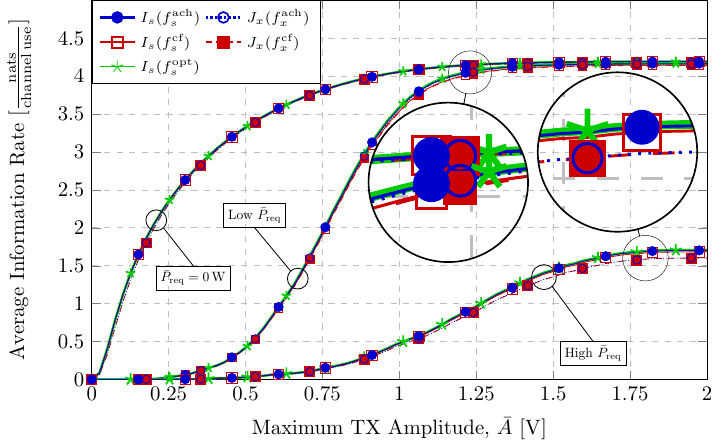} \label{Fig:ResultsMutInf_ImprovedVb}}%
	\caption{Average mutual information and average achievable rate for different input distributions, required average harvested powers $\bar{P}_\text{req}$, and maximum transmit signal amplitudes $\bar{ A }$.}
	\label{Fig:ResultsMutInf}\vspace*{-5pt}
\end{figure*}

First, we observe from Fig.~\ref{Fig:ResultsMutInf} that due to their higher EH efficiencies, the RX circuits equipped with the improved RTD designs significantly outperform the EH circuits employing the RTD design from \cite{Clochiatti2022}.
Next, we note that for pure wireless information transfer with $\bar{P}_\text{req} = 0$, as expected, both proposed pdfs $f_s^\text{ach}$ and $f_s^\text{cf}$ yield the maximum achievable information rate $J^*_{\bar{\mathcal{F}}_x}$.
{ Furthermore, although the proposed pdfs $f_s^\text{ach}$ and $f_s^\text{cf}$ are not optimal, the average maximum mutual information $I_s(f_s^\text{opt})$ is only slightly higher than the average mutual information $I_s(\cdot)$ and average achievable rate $J_x(\cdot)$ obtained with pdfs $f_s^\text{ach}$ and $f_s^\text{cf}$ for all considered values of $\bar{P}_\text{req}$.
Hence, in practical THz SWIPT, one can utilize the low-complexity closed-form input pdf $f_s^\text{cf}$ obtained in Section IV-C without substantial performance degradation.
Moreover, the corresponding achievable information rate can be determined as in Proposition 5.
Since there is no tradeoff between the average achievable rate and the average harvested power for low power ratios $\frac{\bar{P}_\text{req}}{\bar{P}_\text{max}} \leq \frac{1}{3}$, as highlighted in Proposition~\ref{Prop:UniformDistribution}, we observe in Fig.~\ref{Fig:ResultsMutInf} that the average achievable rate and average mutual information obtained for low values of $\bar{P}_\text{req}$ converge to the same value for high maximum TX amplitudes $\bar{A}$.}
Finally, we note that for small TX amplitudes $\bar{A}$ and large required harvested powers $\bar{P}_\text{req}$, i.e., for large power ratios $\frac{\bar{P}_\text{req}}{\bar{P}_\text{max}}$, the values of $I_s(\cdot)$ and $J_x(\cdot)$ depend on both $\bar{A}$ and $\bar{P}_\text{req}$.
Hence, in this case, there is a tradeoff between the average mutual information and the average harvested power that will be investigated in the next section.

\subsection{Optimal and Achievable Rate-Power Tradeoff}
\label{Section:ResultsRatePowerTradeoff}
{In the following, we analyze the tradeoff between the average information rate and the average harvested power in THz SWIPT systems. 
To this end, similar to Fig.~\ref{Fig:ResultsMutInf}, we model $h_\text{f}$ as a Rician distributed random variable with Rician factor $1$ and plot in Fig.~\ref{Fig:Results_Tradeoffs} the average mutual information, average achievable rate, and average harvested power, which are obtained for different values of the maximum amplitude of the transmitted signal $\bar{ A }$ and different RTD designs and averaged over 1000 \gls{iid} random realizations of $h_\text{f}$.
In particular, for different values of the required average harvested power $\bar{P}_\text{req}$, we calculate the average mutual information $I_s(\cdot)$ and average achievable rate $J_x(\cdot)$ by averaging the results obtained for the optimal input pdfs $f_s^\text{opt}$ and the proposed suboptimal pdf $f_s^\text{ach}$ and $f_s^\text{cf}$.}
To confirm the importance of accurate TX signal design for efficient THz SWIPT, for Baseline Scheme 1, we adopt centered Gaussian transmit symbols with pdf $f_s^\text{bs1}(s) = \frac{ \mathcal{N}(\frac{2s - \bar{ A }}{2\sigma_\text{s}}) }{ \sigma_s \text{Erf}(\frac{\bar{ A }}{2\sqrt{2}\sigma_\text{s}}) } $ with support $[0, \bar{ A }]$, where $\mathcal{N}(x) = \frac{1}{\sqrt{2\pi}}\exp(\frac{-x^2}{2})$ is the pdf of the standard Gaussian distribution and $\text{Erf}(\cdot)$ is the error function \cite{Grover2010}.
In particular, for a given value of $\bar{ A }$, the corresponding values of the average mutual information $I_s(f^\text{bs}_s)$ and the average harvested power $\bar{P}_s(f^\text{bs}_s)$ are obtained by adjusting the variance $\sigma_\text{s}^2$ of the truncated Gaussian $f^\text{bs}_s$.
Furthermore, to investigate the impact of accurate EH modelling on the performance of THz SWIPT, for Baseline Scheme 2, we adopt TX signal pdf $f_s^\text{bs2}$, which is obtained by solving optimization problem (\ref{Eqn:GeneralOptimizationProblem}) for a linear EH model, as in \cite{Varshney2008, Zhang2013, Rong2017, Pan2022}.

{First, as in Fig.~\ref{Fig:ResultsMutInf}, we observe in Fig.~\ref{Fig:Results_Tradeoffs} that for all considered maximum transmit signal amplitudes $\bar{ A }$, the improved RTD designs yield higher average information rates than the RTD design from \cite{Clochiatti2022}.
Next, we note that for all considered required average harvested powers $\bar{P}_\text{req}$ and RTD designs, the proposed SWIPT system is able to achieve a significantly higher average mutual information than Baseline Scheme 1, and thus, Gaussian signals are highly suboptimal for THz SWIPT.
Furthermore, we note that for low transmit signal amplitude $\bar{A} = \SI{0.3}{\volt}$, the proposed transmit signal distributions outperform Baseline Scheme 2, which assumes a linear EH model, by a small margin.
However, for large values $\bar{A} \geq \SI{0.75}{\volt}$, the performance gap between the proposed THz SWIPT system and Baseline Scheme 2 is large since the linear EH model cannot capture the non-linearity and non-monotonicity of RTD-based EH circuits.}

As in Fig.~\ref{Fig:ResultsMutInf}, we observe also in Fig.~\ref{Fig:Results_Tradeoffs} that for all ${ A }$ and $\bar{P}_\text{req}$, the gaps between average mutual information $I_s(f_s^\text{opt})$, $I_s(f_s^\text{ach})$, and $I_s(f_s^\text{cf})$ and average achievable rates $J_x(f_x^\text{ach})$ and $J_x(f_x^\text{cf})$ are small.
{Since the average achievable rate $J_x(\cdot)$ and the average mutual information $I_s(\cdot)$ decrease as the required average harvested power increases when $\bar{P}_\text{req} > \frac{1}{3}\bar{P}_\text{max}$, as shown in Fig.~\ref{Fig:ResultsMutInf}, for high required harvested powers, there exists a tradeoff between the average mutual information $I_s(\cdot)$ or average achievable information rate $J_x(\cdot)$ and the average harvested power $\bar{P}_{s}(\cdot)$, respectively, that characterizes the information rate-harvested power regions in Fig.~\ref{Fig:Results_Tradeoffs}.
Thus, in THz SWIPT systems, average mutual information between the TX and RX can be traded for a higher average harvested power at the RTD-based RX.}
Moreover, we note that for low maximum transmit signal amplitudes, which satisfy ${ \bar{A} } < \frac{\sqrt{\rho_1}}{|h|}$, both the average mutual information and the average harvested power grow with ${ \bar{A} }$.
However, if the transmit signal amplitude satisfies $ { \bar{A} }~\geq~ \frac{\sqrt{\rho_1}}{|h|}$, i.e., for ${ \bar{A} } \in \{\SI{0.75}{\volt}, \SI{1}{\volt}\}$ in Fig.~\ref{Fig:Results_Tradeoffs}, the increase of $\bar{A}$ does not yield higher average information rates and average harvested powers for the RTD design from \cite{Clochiatti2022}.
In this case, the tradeoff between the average mutual information and average harvested power of the THz SWIPT system is determined by the maximum instantaneous harvested power $\psi(\rho_1)$ and not by the value of ${ \bar{A} }$.
On the contrary, for the RTD with improved $U_\text{br}$, whose $\rho_1$ and $\rho_\text{max}$ are higher, the average mutual information and average harvested power grow for $\bar{A} > 0.75$.
{Thus, we conclude that an accurate EH model and a corresponding TX signal design are essential for efficient THz SWIPT.}

\begin{figure*}[!t]
	\centering
	\subfigure[RTD design from \cite{Clochiatti2022}]{
		\includegraphics[draft=false, width=0.48\textwidth]{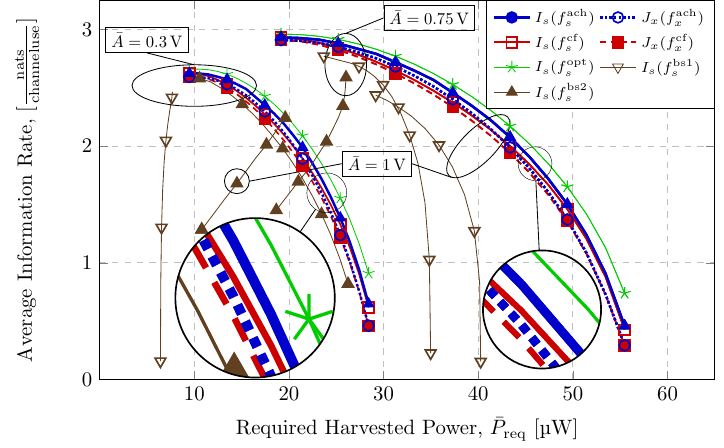} \label{Fig:Results_Tradeoffs_Realistic}}
	\subfigure[Improved RTD design with lower $I_\text{rev}$]{
		\includegraphics[draft=false, width=0.48\textwidth]{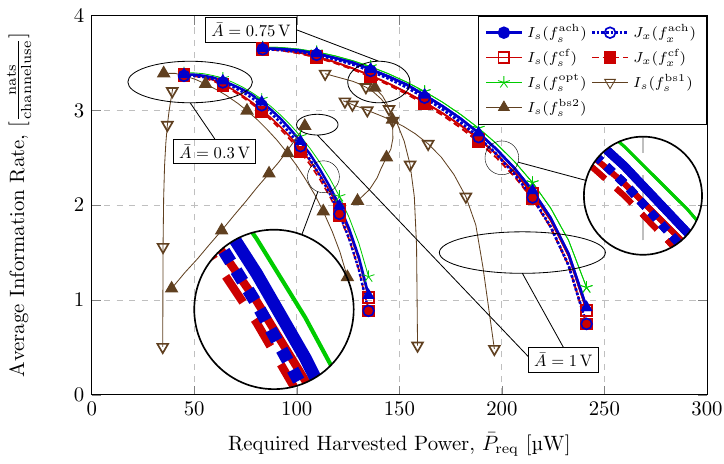} \label{Fig:Results_Tradeoffs_ImprovedId}}\hspace*{5pt}
	\subfigure[Improved RTD design with higher $U_\text{br}$]{
		\includegraphics[draft=false, width=0.48\textwidth]{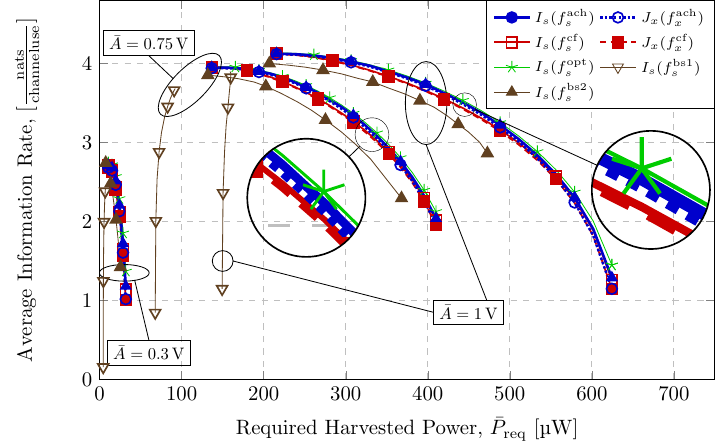} \label{Fig:Results_Tradeoffs_ImprovedVb}}%
	\caption{Optimal and achievable average information rate-average harvested power regions for different values of $\bar{A}$ and different RTD designs.}
	\label{Fig:Results_Tradeoffs}
\end{figure*}

\section{Conclusions}
In this work, we studied THz SWIPT systems, where the RXs were equipped with RTD-based EH circuits to extract both information and power from the received unipolar ASK-modulated signal.
To characterize the instantaneous output power at the RX, we proposed a general non-linear piecewise EH model, whose parameters were tuned to fit circuit simulation results.
Furthermore, we derived the transmit signal pdf that maximizes the mutual information between TX and RX subject to constraints on the average harvested power at the RX and the peak signal amplitude at the TX.
Since the computational complexity of this pdf may be high in practice, we also derived the input pdfs maximizing the achievable information rate.
In addition, we proposed a closed-form input pdf that yields a suboptimal solution of the problem.
{Our numerical results demonstrated that the proposed novel RTD designs with lower reverse current flow and a higher breakdown voltage are preferable if the input signal power at the RX is low and high, respectively.
Furthermore, we showed that for low and high received powers, the information rate-harvested power tradeoff depends on the peak transmit amplitude and the maximum instantaneous harvested power, respectively.
Finally, we observed that the proposed optimal and both low-complexity suboptimal input pdfs yield similar THz SWIPT performance and significantly outperform two baseline schemes based on non-negative Gaussian transmit signals and optimal distribution for a linear EH model, respectively.
Thus, we conclude that the transmit signal distribution, RX design, and EH model have to be carefully jointly considered to achieve efficient THz SWIPT.}

{Interesting directions for future research on THz SWIPT include the optimal design of multi-user multi-antenna systems employing RTDs \cite{Boshkovska2015, Xu2022}, the modelling of THz TXs \cite{Li2022a, Kazan2021}, the design of transmit signal waveforms that are robust to hardware impairments and imperfect channel knowledge \cite{Boshkovska2018, Zhu2021}, and the analysis of THz SWIPT systems in the near field \cite{Shinohara2021, Mayer2023}.}
	
\appendices
	\renewcommand{\thesection}{\Alph{section}}
	\renewcommand{\thesubsection}{\thesection.\arabic{subsection}}
	\renewcommand{\thesectiondis}[2]{\Alph{section}:}
	\renewcommand{\thesubsectiondis}{\thesection.\arabic{subsection}:}	

	
	
	
	
	
	\section{Proof of Proposition \ref{Prop:OptimalSolution}}
	\label{Appendix:OptimalSolution}
	First, since $x$ is a deterministic function of input signal $s$, we obtain $J^*_{\bar{\mathcal{F}}_{x}}$ for given ${\bar{A}}$ and $\bar{P}^\text{\upshape req}_\text{\upshape harv} \geq \frac{1}{3} \bar{P}_\text{\upshape max}$ as follows:
\begin{equation}
	J^*_{\bar{\mathcal{F}}_{x}} = \max_{f_x \in \bar{\mathcal{F}}_x} \; \frac{1}{2} \ln\bigg(1 + \frac{e^{2 {h}_x(f_x)}}{2\pi e \sigma^2 }\bigg).
	\label{Eqn:Prop3Eqn1} 
\end{equation}
Since function $J_x(\cdot)$ is monotonically increasing in ${h}_x(\cdot)$, for $\bar{P}_\text{\upshape req} \in [\frac{1}{3} \bar{P}_\text{\upshape max}, \bar{P}_\text{\upshape max}]$, the pdf $f_x \in \bar{\mathcal{F}}_x$ solving (\ref{Eqn:Prop3Eqn1}) is the maximum entropy distribution, i.e., $f_x$ maximizes entropy ${h}_x(\cdot)$ among all pdfs $f_x \in \bar{\mathcal{F}}_x$ that satisfy $\mathcal{D}\{f_x\} \in [0, \sqrt{\bar{P}_\text{\upshape max}}]$ and $\mathbb{E}_x\{x^2\} \geq \bar{P}_\text{\upshape req}$ \cite{Lapidoth2009}.
Exploiting the Karush–Kuhn–Tucker (KKT) conditions, it can be shown that the optimal pdf as solution of (\ref{Eqn:Prop3Eqn1}) is given by $f^\text{ach}_x(x)$ in Proposition~\ref{Prop:OptimalSolution}.
Moreover, the corresponding entropy of $x$ is given by
\begin{equation}
	h_x(f_x^\text{ach}) = -\int_{x} f_x^\text{ach}(x) \ln(f_x^\text{ach}(x)) \text{d} x = {\mu}_0 - {\mu}_1^2 \frac{\bar{P}_\text{\upshape req}}{\bar{P}_\text{\upshape max}} .
	\label{Eqn:Prop3Eqn3}
\end{equation}
Finally, substituting (\ref{Eqn:Prop3Eqn3}) into (\ref{Eqn:Prop3Eqn1}), we obtain (\ref{Eqn:OptMI}).
This concludes the proof.
	
	

\bibliographystyle{IEEEtran}
\bibliography{WPT_THz.bib}

\end{document}